\documentclass[11pt]{article}

\usepackage[margin=1in]{geometry}
\usepackage[style=alphabetic,backend=biber,doi=true,maxnames=99]{biblatex}
\bibliography{ref.bib}

\input{arxiv_macros}
\DeclareMathOperator{\poly}{poly}
\newcommand{\OPT}{\mathrm{OPT}}

\newcommand{\cost}{\mathrm{cost}}
\newcommand{\load}{\mathrm{load}}
\newcommand{\loadvec}{\overrightarrow{\load}}
\newcommand{\ADAPTLP}{\mathsf{ADAPT}\text{-}\mathsf{LP}}
\newcommand{\ADAPTCLP}{\mathsf{ADAPT}\text{-}\mathsf{CR}}
\newcommand{\slbr}{\ensuremath{\mathsf{SLBR}}\xspace}

\newcommand{\lb}{\mathrm{lb}}
\newcommand{\denlb}{\mathrm{density}}
\newcommand{\mx}{\mathrm{max}}
\newcommand{\total}{\mathrm{total}}
\newcommand{\nxt}{\mathrm{next}}
\newcommand{\supp}{\mathrm{supp}}
\newcommand{\rank}{\ensuremath{\mathrm{r}}}
\newcommand{\viol}{\nu}
\newcommand{\mksp}{\mathrm{makespan}}
\newcommand{\OPTAD}{\mathrm{OPT_{AD}}}

\newcommand{\mc}{\mathcal}
\newcommand{\M}{\mc{M}}

\newcommand{\Q}{\ensuremath{\mathcal{Q}}}
\newcommand{\Rp}{\ensuremath{\mathbb R}_{\geq 0}}
\newcommand{\R}{\ensuremath{\mathbb R}}
\newcommand{\Z}{\ensuremath{\mathbb Z}}

\renewcommand{\Pr}{\operatorname{\mathbb{P}}}
\newcommand{\E}{\operatorname{\mathbb{E}}}
\newcommand{\ceil}[1]{\ensuremath{\left\lceil#1\right\rceil}}

\newcommand{\norm}[1]{\ensuremath{\Vert #1 \Vert}}

\DeclareMathOperator{\Ber}{Bernoulli}

\newcommand{\sm}{\setminus}
\newcommand{\sg}{\sigma}
\newcommand{\sgp}{{\sg^\prime}}
\newcommand{\sgs}{{\sg^*}}
\newcommand{\ld}{\lambda}

\newcommand{\bz}{\overline{z}}
\newcommand{\by}{\overline{y}}

\newcommand{\bone}{\mathds{1}}
\newcommand{\euler}{\ensuremath{\mathrm{e}}}
\newcommand{\LHS}{\ensuremath{\mathrm{LHS}}}

\newcommand{\s}{\ensuremath{S}} 
\newcommand{\I}{\ensuremath{\mc{I}}}
\newcommand{\F}{\ensuremath{\mc{F}}}
\newcommand{\U}{\ensuremath{\mc{U}}}
\newcommand{\tr}{\mathrm{tr}}
\newcommand{\ex}{\mathrm{ex}}
\newcommand{\A}{\ensuremath{\mc{A}}}

\newcommand{\smooth}{smooth\xspace}
\newcommand{\smoothing}{smoothing\xspace}
\newcommand{\smoothed}{smoothed\xspace}

\newcommand{\consistent}{consistent\xspace}

\newcommand{\subsets}[1]{\ensuremath{2^{#1} \sm \{ \emptyset \}}} 

\newcommand{\sse}{\subseteq}
\newcommand{\jobset}{\mc{J}}

\newcommand{\asgn}{\pi}

\newcommand{\mset}{\sg}

\newcommand{\delout}{\ensuremath{\delta^{\mathrm{out}}}}
\newcommand{\scaleparm}{\Delta}
\newcommand{\indexset}{\mathcal{K}}
\newcommand{\sml}{\mathsf{sml}}
\newcommand{\med}{\mathsf{med}}
\newcommand{\lrg}{\mathsf{lrg}}

\title{Stochastic Load Balancing with Machine Reservations%
\thanks{A preliminary version of this paper appeared in the proceedings of the
27th International Conference on Integer Programming and Combinatorial
Optimization, 2026.}}

\author{%
David Alem\'an Espinosa%
\thanks{University of Waterloo, Waterloo, Canada.
Email: \texttt{dalemanespinosa@uwaterloo.ca}.
Supported in part by C.~Swamy's NSERC Discovery Grant.}
\qquad
Naveen Garg%
\thanks{Indian Institute of Technology Delhi, New Delhi, India.
Email: \texttt{naveen@cse.iitd.ac.in}.
Research was supported by the Usha Hasteer Chair.}
\qquad
Sharat Ibrahimpur%
\thanks{Eidgen\"ossische Technische Hochschule Z\"urich, Z\"urich, Switzerland.
Email: \texttt{sibrahimpur@ethz.ch}.
Part of this work was conducted while the author was a postdoctoral researcher
at the London School of Economics and Political Science and the University of
Bonn. His research was supported by the European Union's Horizon 2020 research
and innovation programme under grant agreement No.~757481 (ScaleOpt), the Dutch
Research Council NWO Vidi Grant 016.Vidi.189.087, and the Deutsche
Forschungsgemeinschaft (DFG, German Research Foundation), project
No.~537750605.}
\\[1ex]
Neil Olver%
\thanks{London School of Economics and Political Science, London, United Kingdom.
Email: \texttt{n.olver@lse.ac.uk}.}
\qquad
Chaitanya Swamy%
\thanks{University of Waterloo, Waterloo, Canada.
Email: \texttt{cswamy@uwaterloo.ca}. Research supported in part by an NSERC Discovery Grant.}
}

\date{}

\begin{document}

\maketitle

\begin{abstract}
We introduce a novel variant of stochastic load balancing that enables a quantitative tradeoff between the practical benefits of non-adaptive policies and their performance limitations. 
Our model describes a solution in two stages. 
In the first stage, given only job-size distributions, we reserve a set of at most $k$ machines for each job (a $k$-reservation).
In the second stage, after observing job-size realizations, we assign each job to one of its reserved machines (a consistent assignment). 
The goal is to minimize the expected makespan. 
If $k=1$, we get the standard stochastic load balancing problem of finding a non-adaptive assignment with minimum expected makespan. 
If $k$ is equal to the number of machines, then we obtain an all-powerful omniscient optimum that can tailor the assignment arbitrarily to the job-size realizations. 

\medskip

We give a number of results that quantify this tradeoff.
Most saliently, we show that in the setting of identical machines, a $2$-reservation suffices to achieve a constant-factor approximation to the omniscient optimum, establishing a ``power-of-two-choices'' result for stochastic load balancing.
We also show that this no longer holds true in the more challenging setting of related machines.
Nonetheless, we give a number of positive algorithmic results for this setting: a true $O(\log m/\log \log m)$-approximation; a bicriteria $O(1)$-approximation by reserving twice as many machines per job relative to an optimal $k$-reservation; and a $2$-reservation whose cost is within a constant factor of what the \emph{adaptive} optimum can achieve.

\end{abstract}


\clearpage

\section{Introduction} \label{sec:intro}

Uncertainty is an unavoidable facet of various decision environments, and
a thriving area of optimization, {\em stochastic optimization}, deals with optimization under uncertainty. 
Stochastic load-balancing and scheduling problems constitute a prominent and well-studied class of stochastic-optimization problems 
(see, e.g.,~\cite{EberleGMMZ25,GoelI99,GuptaKNS21,GuptaMUX20,IbrahimpurS20,ImMP15,KleinbergRT00,MoehringSU99,Pinedo04}). 
The typical setup in {\em stochastic load balancing} involves some $m$ machines, and job sizes (a.k.a. processing times) that are random variables whose distributions are specified in the input. 
The goal is to come up with a job-to-machine assignment that minimizes the expected makespan, where {\em makespan} is the maximum load on a machine. 
Two ways of computing and specifying the assignment are considered in the literature, based on what information is available for the computation: a {\em non-adaptive} policy computes the entire job-to-machine assignment given only the distributional information and without knowing the job-size realizations, while an {\em adaptive} policy
can make sequential decisions: it can inspect jobs sequentially and choose the assignment of the currently inspected job based on the realizations of previously inspected jobs.  

An appealing aspect of a non-adaptive policy is that it is easy to implement and specify: since the uncertainty does not influence decisions at ``run-time'', i.e., when job sizes become known, there is little (logistical or computational) overhead incurred at run-time. 
Moreover, in certain settings, the certainty of knowing which machine will process which job can also be desirable for other compelling, non-computational reasons.
For instance, consider a bank or investment firm that has to decide how to assign fund managers to cater to its clients. 
Here the machines correspond to fund managers, and jobs correspond to clients. 
The amount of time needed to serve a client is subject to
uncertainty. 
Taking a client-centric perspective, the firm would like to assign each client
to a specific, dedicated fund manager (i.e., come up with a non-adaptive
policy). 
This could offer the client a more personalized experience, enhance client satisfaction, cultivate customer loyalty, and help build a long-term relationship with the client. 
Non-adaptive policies are also often easier to design and analyze than adaptive policies, partly because an adaptive policy could lead to an exponential-size decision tree.
On the flip side, being more constrained, even an optimal non-adaptive policy may incur a much larger cost than an optimal adaptive policy; that is, the non-adaptive optimum may be much larger than the adaptive optimum. 
Concretely, for stochastic load balancing: (1)~the {\em adaptivity gap}, which measures the worst-case ratio of the non-adaptive optimum to the adaptive optimum, over all instances, is $\Omega(\log m/\log\log m)$~\cite{GuptaKNS21}, even for $m$ identical machines\footnotemark;
\addtocounter{footnote}{-1}
(2)~one can devise non-adaptive policies that achieve an $O(1)$-approximation relative to the non-adaptive optimum~\cite{GuptaKNS21,KleinbergRT00}, even on unrelated
machines%
\footnote{The machine environments, identical, related, and unrelated, are defined as follows. In general, the processing time of a job can depend on the machine to which it is assigned. With {\em unrelated} machines, the processing-time vectors $\{p_{ij}\}_{i\in[m]}$ and $\{p_{ij^\prime}\}_{i\in[m]}$ of two jobs $j$ and $j^\prime$ need not have any relationship. With {\em related} machines, each machine $i$ has a
speed $s_i>0$, and we have $p_{ij}=p_j/s_i$ for every machine $i$ and job $j$; with identical machines, we have $p_{ij}=p_j$ for every $i,j$.}
\addtocounter{footnote}{-1}
(and even for more general objectives~\cite{IbrahimpurS20,Molinaro19}); and  
(3)~designing an adaptive policy that achieves a better guarantee relative to the adaptive optimum than a non-adaptive policy is open for unrelated machines, and only recently an $O(1)$-approximate adaptive policy was obtained for related machines\footnotemark~\cite{EberleGMMZ25}.

\subsection{Our Contributions} \label{contrib}
In this work, we propose and study a novel model for stochastic load balancing that allows one to quantitatively trade off the attractive features of non-adaptive policies and the associated performance loss resulting from their constrained nature.
To motivate our model, let us reconsider the investment-firm example discussed above. 
For simplicity, suppose that the uncertainty in client service times is such that each client arrives with some probability, and if it arrives, then it uses up a fixed amount of service time. 
At the beginning of the day, the firm gets to know which clients will be arriving during the day, and has to assign these clients to
fund managers. 
Suppose that there is no way of choosing dedicated managers for each client to obtain reasonably balanced loads for the managers.%
\footnote{The $\Omega(\log m/\log\log m)$ adaptivity gap holds even with
such weighted Bernoulli job-size distributions~\cite{GuptaKNS21}.}
As a means of alleviating this situation, the firm may consider incorporating a little flexibility in assigning clients to managers, in the following way. 
The firm now reserves a small dedicated {\em set of managers} for each client, and always assigns a client to one of the managers reserved for them. 
This flexibility can enable the firm to choose, at the beginning of the day
when the set of arriving clients becomes known, an assignment that meets the
above consistency constraint and better balances the loads. Provided that the set of dedicated managers for each client is ``small,'' the firm can retain a similar flavor of personalized client service. 
The size of the reserved
manager set therefore controls the trade-off between the benefits of
dedicated personalized service and the quality of load balancing. 

\subsubsection{Problem Description.}
We formalize the above type of setup via the following two-stage model, which we dub {\em two-stage stochastic load balancing with reservations} (\slbr). 
We have a machine set $M = [m]$ and a job set $\jobset$. 
In the first stage, given only the job-size distributions, we reserve, for each job $j$, a set $\mset(j)\sse[m]$ of machines with $1 \leq |\mset(j)| \leq k$, where $k$ is an input parameter. 
In the second stage, we observe the job-size realizations, and have to assign each job $j$ to a machine $\asgn(j)\in\mset(j)$. 
We refer to $\mset$ as a {\em $k$-reservation}, and say that $\asgn:\jobset\to[m]$ is a {\em $\mset$-\consistent assignment}. 
The goal is to minimize the expected makespan.  
Throughout, we use ``reservation'' to denote the first-stage decision, and ``assignment'' to denote the assignment of jobs to machines in the second stage (\consistent with the
first-stage reservation). 

Formally, given job sizes $\{p_{ij}\}_{i\in[m],j\in\jobset}$ and an assignment $\asgn : \jobset \to M$, the load on a machine $i \in M$ is $\load(i;\asgn,p) \coloneqq \sum_{j:\asgn(j)=i} p_{ij}$. 
The cost of a $k$-reservation $\mset$ for these job sizes is the minimum
makespan achievable by a $\mset$-\consistent assignment: 
\begin{equation*} \label{eq:secondstageobj} 
  \cost(\mset;p) \coloneqq \min_{\substack{\mset\text{-\consistent} \\ \text{assignment } \asgn}}
  \mksp(\asgn;p),
\end{equation*} 
where $\mksp(\asgn;p) \coloneqq \max_{i \in [m]} \load(i;\asgn,p)$.
In an instance of stochastic load balancing with reservations, the vector $\{P_{ij}\}_{i\in[m]}$ of processing times of a job $j$ on the different machines is a random vector, whose distribution is specified in the input.
We often consider the setting where jobs are independent, which means that random variables $P_{ij}$ and $P_{i^\prime j^\prime}$ are independent whenever $j\neq j^\prime$.
The goal is to choose a $k$-reservation $\mset$ to minimize
$\cost(\mset) \coloneqq \E[\cost(\mset;P)]$, where the expectation is taken over the realizations of the $P_{ij}$s. 
This corresponds to computing a good first-stage decision. We would also, of course, like to compute a $\mset$-\consistent assignment in the second stage, for each realization of the $P_{ij}$s. This becomes an instance of {\em deterministic} makespan-minimization on unrelated machines, which admits well-known $2$-approximation algorithms~\cite{LenstraST90,ShmoysT93}. We therefore focus on the first-stage computation. 

We use $\OPT_k(\I)$ to denote the optimal value for instance $\I$; we drop $\I$ when this is clear from the context. 
When $k=1$, the only $\mset$-\consistent assignment is $\mset$ itself, so the problem becomes that of computing the best non-adaptive assignment. 
Thus, $\OPT_1$ is the non-adaptive optimum. 
At the other extreme, when $k=m$, taking $\mset(j)=[m]$ for all $j$ imposes no restrictions on the second-stage assignment; so $\OPT_m=\E[\min_{\text{assignment }\asgn}\mksp(\asgn;P)]$, which we refer to as the {\em omniscient optimum}.
As $k$ increases, we increase the flexibility allowed for the second-stage assignment, and so $\OPT_k$ decreases with $k$. 
(Note that with deterministic jobs, we have $\OPT_1=\OPT_m$---we can take the optimal assignment as the $1$-reservation---so the setting
$k>1$ is relevant only with stochastic jobs.)  

\smallskip

Various interesting lines of inquiry arise in our model.
\vspace{-5pt}
\begin{enumerate}[label=$\bullet$]
\item First, there is the algorithmic question of designing good approximation algorithms for the problem. 
As noted above, \slbr with $k=1$ corresponds to obtaining a good approximation relative to the non-adaptive optimum considered in prior work~\cite{GuptaKNS21,KleinbergRT00}, so we consider $k\geq 2$. 

\smallskip

\item Another pertinent question that we explore is: {\em how much do we stand to lose by constraining ourselves to use \consistent (second-stage) assignments?} 
Since $\OPT_m$ and $\OPT_k$ are, respectively, the best expected makespan one can achieve without any restrictions and with the restriction that the assignment be \consistent with a $k$-reservation, this is formally the {\em analysis} question of determining the ratio $\OPT_k(\I)/\OPT_m(\I)$ for a specific instance $\I$, or its worst-case value across all instances. 
We coin the term {\em price of $k$-reservation} to refer to the worst-case ratio. 
Note that this quantity decreases with $k$, and the price of $m$-reservation is $1$ by definition.

A low price of $k$-reservation for a small value of $k$ has the appealing interpretation that a little flexibility goes a long way. 
Instead of fixing the job-to-machine assignment beforehand, if we only require that each job be assigned to one of $k$ preselected machines for that job, then one can alleviate the performance loss faced by non-adaptive assignments and get close in performance to a solution that is allowed to arbitrarily tailor the job-assignment in the second stage to the job-size realizations. 

\smallskip

\item Our two-stage reservation-based model and adaptive policies can both be viewed as enhancements of non-adaptive policies geared toward obtaining better-quality solutions, albeit of a very different nature. 
An adaptive policy has {\em partial information} about job-size realizations when assigning jobs, knowing the realizations of previously assigned jobs, but the job-assignment is {\em unconstrained}. 
In a reservation-based algorithm, by contrast, the assignment must be {\em \consistent with the reservation}, which is chosen before any job sizes are realized, but one has {\em full information} about job-size realizations when determining the assignment. 

Although this is somewhat of an apples-to-oranges comparison, we can ask about the relative effectiveness, or power, of these two models, by analyzing the ratio of $\OPT_k$ and the adaptive optimum, $\OPTAD$. 
This ratio is always bounded by the price of $k$-reservation; but are there settings where the price of $k$-reservation is large, yet $\OPT_k/\OPTAD$ is small, showing that $k$-reservations are nevertheless as powerful as adaptive policies, loosely speaking? 
If $\OPT_k/\OPTAD$ is small (say $O(1)$) for a small value of $k$ (e.g., $k=2$), then the upshot is that one can retain, to an extent, the benefits of a non-adaptive assignment, and at the same time reap the performance gains of an adaptive policy. 
\end{enumerate}

\subsubsection{Our Results.}
We obtain especially crisp answers to the above questions for the setting of identical (a.k.a. parallel) machines, where $P_{ij}=P_j$ for every machine $i$ and job $j$. 
Strikingly, we show that $2$-reservations are already powerful enough to yield an $O(1)$-approximation to $\OPT_m$. 

\begin{theorem} \label{thm:iden2options}
For any $\slbr$ instance on identical machines with $k \geq 2$, we can efficiently compute a random $2$-reservation $\mset$ such that $\E[\cost(\mset)] = O(\OPT_m)$. 
\end{theorem}
Thus, we simultaneously obtain {\em an $O(1)$-approximation algorithm, for any $k$}, and that {\em the price of $k$-reservation is $O(1)$, even for $k=2$}.\footnote{The price of $1$-reservation for $m$ identical machines is $\Theta(\log m/\log \log m)$. The upper bound follows from reserving a single machine independently and uniformly for each job. The lower bound is realized by an instance with $m^2$ i.i.d. $\Ber(1/m)$ jobs for which $\OPT_m$ and $\OPTAD$ are both $\Theta(1)$ while $\OPT_1$ is $\Theta(\log m/\log \log m)$.} 
The latter statement can be viewed as a ``power-of-two-choices result'' for stochastic load balancing, showing that allowing even $2$ preselected machines enables one to obtain substantially smaller expected makespan.
The guarantee in \Cref{thm:iden2options} in fact holds for a random $2$-reservation which can be computed without knowing the underlying distributions, and so this guarantee holds even with {\em arbitrarily correlated jobs}. 
Note that the latter setting can model an oblivious adversary, who can choose arbitrary job sizes but does not know the random bits used to compute the $2$-reservation. 
\Cref{thm:iden2options} is proved in \Cref{sec:identical}. 

\smallskip

We next consider the more general setting with related machines. 
Each machine $i \in M$ has a speed $s_i > 0$, and each job $j \in \jobset$ has a random size $P_j$. 
The processing time of job $j$ on machine $i$ is $P_{ij} = P_j/s_i$.  
All our results for related machines assume that the $P_j$s are mutually independent. 
In stark contrast with identical machines, the price of $k$-reservation is no longer bounded by a constant for any constant value of $k$. 

\begin{restatable}{theorem}{relporlbthm} \label{relporlb-thm} 
For every positive integer $k$ and every $m = \exp(\Omega(k \log k))$, there is an instance of $\slbr$ on $m$ related machines for which $\OPT_k/\OPT_m = \Omega(\sqrt[2k]{m}/k)$. 
In particular, for any $k=o(\log m/\log\log m)$, the price of $k$-reservation is unbounded for sufficiently large $m$.
\end{restatable}

Complementing this, we show that the price of $k$-reservation is a constant already when $k \approx \log m$. 

\begin{restatable}{theorem}{relporubthm} \label{relporub-thm} 
For any instance of $\slbr$ on $m$ related machines, the price of $\Omega(\log m)$-reservation is $O(1)$. 
\end{restatable}

The above two theorems are proved in \Cref{sec:related-por}.  
Next, we turn to designing approximation algorithms.  Related machines pose substantial technical challenges. 
In \Cref{sec:overview}, we illustrate with a simple example how two simple strategies fail to yield constant-factor approximations for \slbr on related machines. 
We overcome some of these challenges by developing suitable machinery for obtaining three different flavors of approximation guarantees. 
First, we use a linear-programming-based approach combined with standard concentration inequalities to obtain a true approximation guarantee. 

\begin{restatable}{theorem}{rellogthm} \label{rellog-thm} 
There is an $O(\log m/\log \log m)$-approximation algorithm for instances of \slbr with $m$ related machines. 
\end{restatable}

In the second approach, we allow ourselves twice as many reservations per job and obtain a bicriteria approximation. 

\begin{restatable}{theorem}{relbicriteriathm} \label{relbicriteria-thm} 
For any \slbr instance on $m$ related machines with $k \geq 1$, we can efficiently compute a random $2k$-reservation $\mset$ with $\E_{\mset}[\cost(\mset)] = O(\OPT_k)$. 
\end{restatable}

Finally, we showcase the effectiveness of $2$-reservations versus adaptive policies. 

\begin{restatable}{theorem}{reladapthm} \label{reladap-thm} 
For any \slbr instance on $m$ related machines with $k \geq 2$, we can efficiently compute a random $2$-reservation whose expected cost is within constant factors of the adaptive optimum. 
\end{restatable}

Thus, while $k$-reservations cannot get us close to the omniscient optimum for any constant $k$, even $2$-reservations are as powerful as adaptive policies (modulo $O(1)$ factors).\footnote{In fact, $2$-reservations can sometimes do even better than adaptive policies: we exhibit an instance (\Cref{lem:adaptvstworeslb}) where the adaptive optimum exceeds the cost of an optimal $2$-reservation by a factor of $\Omega(\sqrt{m})$.} As a consequence, we obtain that for weighted Bernoulli job-size distributions---an oft-considered special case (see, e.g., \cite{AntoniadisHSU25,EberleFMM19,GuptaMZ23,KleinbergRT00}) that sometimes encapsulates the crux of the general setting~\cite{GuptaKNS21,KleinbergRT00}---the price of $2$-reservation is $O(1)$. This is because one can show that, with weighted Bernoulli distributions, the adaptive optimum is $O(\OPT_m)$.

Our algorithmic results for $\slbr$ on related machines (\Cref{rellog-thm,relbicriteria-thm,reladap-thm}) can be found in \Cref{sec:related,sec:adaptive}. 

\subsection{Related work} \label{relwork}

As noted earlier, taking $k=1$ in our model, we obtain the classical stochastic load-balancing problem of finding a non-adaptive assignment with minimum expected makespan, which is a problem with a rich history of study. 
In their seminal work, Kleinberg, Rabani and Tardos~\cite{KleinbergRT00} developed an 
$O(1)$-approximation algorithm for this problem on identical machines. 
They were the first to study stochastic load balancing with arbitrary job-size distributions, whereas prior work had only considered specialized distributions (see, e.g.,~\cite{BrunoDF81,Weiss95,WeissP80}). 
Recently, Gupta, Kumar, Nagarajan and Shen~\cite{GuptaKNS21} obtained an $O(1)$-approximation for the general setting of unrelated machines, and subsequently, such results were also obtained for more general norm-based objectives of the machine-load vector~\cite{IbrahimpurS20,IbrahimpurS22,Molinaro19}.  
Another stream of work has focused on obtaining improved approximation guarantees for specialized distributions, such as exponential and Poisson distributions~\cite{DeKLN20,GoelI99,IbrahimpurS21}.    

In contrast to the above line of work, which focuses on obtaining guarantees relative to the non-adaptive optimum, the study of (adaptive) policies that obtain guarantees relative to the {\em adaptive optimum} has received very little attention. 
Gupta et al.~\cite{GuptaKNS21} showed that non-adaptive policies are quite limited in this respect: the adaptivity gap is $\Omega(\log m/\log\log m)$, even for identical machines and weighted Bernoulli job-size distributions. 
Very recently, Eberle, Gupta, Megow, Moseley and Zhou~\cite{EberleGMMZ25} devised an $O(1)$-approximation {\em adaptive policy} for the setting of related machines, showing for the first time how adaptive policies can do better than non-adaptive policies. 
They also showed that the $\Omega(\log m/\log\log m)$ lower bound on the adaptivity gap is tight by designing a non-adaptive policy with a matching approximation guarantee, even for the setting of unrelated machines. 
We use some of their techniques in obtaining our approximation results for related machines. Obtaining an adaptive policy for unrelated machines that does better remains a challenging open question.

The high-level idea of our model---that providing some limited choice in deciding the job assignment at a second stage may allow for dramatically better solutions---invites comparison to the rich stream of literature on the ``power of two choices''~\cite{MRS01}.  
The prototypical result in this area is the balls-and-bins result of~\cite{AzarBKU99}, which shows that with $m$ identical machines and $m$ deterministic unit-size online jobs, if we choose $d$ random machines for each job, and assign the job to the least-loaded machine among these $d$ machines, then the expected makespan decreases from $\Theta(\log m/\log\log m)$ when $d=1$, to $O(\log\log m/\log d)$ when $d\geq 2$. They also proved that the expected {\em offline} optimum is $O(1)$, i.e., if we get to see all jobs and can then assign jobs, but in a way that is \consistent with the machine-set chosen for each job, then the expected makespan becomes $O(1)$. 
The latter result has the same flavor as our result for identical machines (Theorem~\ref{thm:iden2options}): indeed, using our terminology, they show that the cost of a random $2$-reservation is $O(1)$, albeit in the extremely special case of {\em identical} jobs. 

\slbr is an example of a two-stage stochastic optimization problem with recourse.
Two-stage stochastic optimization problems have been extensively studied in the Operations Research literature, from a modeling and computational viewpoint, and in the theoretical CS (TCS) literature from an approximation-algorithms viewpoint; see, e.g., the text~\cite{BirgeL97}, and the survey~\cite{SwamyS06}, which discusses work in TCS in this area. 
While work on stochastic load balancing has focused exclusively on independent jobs, work in two-stage stochastic optimization has also considered other, more general distribution models that incorporate arbitrary correlation effects.  
These include the scenario model, where the input specifies all realizations of the uncertain data, and the most-general black-box model, where one only assumes that one has sampling access to the underlying distribution; approximation algorithms have been devised even in the black-box model~\cite{GuptaPRS04,ShmoysS06}. 
We note that in order to obtain these positive results in the black-box model, an important property that the two-stage problem needs to satisfy is that every first-stage action has a corresponding (possibly more-expensive) second-stage action; \slbr does not satisfy this property.

Before concluding this section, we mention two other models of stochastic scheduling that involve a notion of limited adaptivity distinct from ours. 
The first model is due to Sagnol and Schmidt genannt Waldschmidt~\cite{SagnolSchmidtGW21}, who studied makespan minimization for independent stochastic jobs on $m$ identical machines.
At time $0$, using only the job-size distributions, each job is assigned to a machine and given a position in that machine's queue.
Their model is parameterized by a delay $\delta>0$ and a shift interval $\tau>0$.
At any time $t$ that is a multiple of $\tau$, a job $j$ that has not yet begun processing may be moved to the queue of another machine, but it cannot begin processing before time $t+\delta$.
Thus, their policies adapt sequentially during execution by modifying the queues in response to the processing times revealed so far.
In contrast, our model makes no assignments during execution: it reserves a set of machines for each job before the job sizes are revealed and, after observing the entire realization, assigns each job to one of its reserved machines.
The second model is due to Chen, Megow, Rischke and Stougie~\cite{ChenMRS15}, who considered a different two-stage reservation model.
Time slots can be reserved at a relatively low cost before the uncertain processing times are known.
After the processing times are revealed, additional slots may be purchased at a higher second-stage cost if the capacity reserved initially is insufficient to process all jobs.

\section{Technical Overview} \label{sec:overview}

In this section, we give an overview of our techniques, discussing the main ideas underlying our algorithms and their analysis. 
A full proof of our result on identical machines can be found in \Cref{sec:identical}. 
Due to limited space, detailed proofs of our results on related machines have been deferred to the full version of this paper. 

\subsection{Effectiveness of \texorpdfstring{\boldmath $2$}{2}-Reservations versus Omniscient Optimum for Identical Machines}

Recall that for identical machines, we obtain the rather strong positive result that $2$-reservations are sufficient to get within a constant factor of the omniscient optimum $\OPT_m$. 
This holds even with arbitrarily correlated job sizes.  
Our algorithm here is very simple: the reservation for each job simply consists of $2$ machines chosen independently and uniformly at random.  
This algorithm and its analysis also lie at the heart of our results for related machines.

One of the main ingredients for analyzing this algorithm is a {\em density-based approximation} of the cost of a given reservation $\sigma$ (see parts
(i) and (ii) of Lemma~\ref{lem:density}). 
We can view a $2$-reservation $\mset$ as specifying a graph, where we have a node for each machine, and, for each job $j$, we have an edge (possibly a self-loop) joining the machines in $\mset(j)$. 
A $\mset$-\consistent assignment corresponds to an orientation of this graph. 
For given job sizes $\{p_j\}_{j\in\jobset}$, the makespan of an orientation is the maximum weighted in-degree of a machine node, where the edge corresponding to $\mset(j)$ has weight $p_j$.  
We consider a {\em fractional} relaxation of the problem of finding an orientation to minimize the maximum weighted in-degree. 
This can be easily cast as a flow problem. Its dual is a cut problem that can be interpreted as finding a maximum-density node-set in this graph, where the density of a set $R\sse[m]$ is ${\text{(total weight of edges contained in $R$)}}\,/\,{|R|}$. 
Also, one can argue that a fractional orientation can be rounded to a $\mset$-\consistent assignment of makespan at most twice the maximum density. 
Thus, our goal is to find a $2$-reservation $\mset$ with small maximum density.

We prove that, for any job sizes, a random $2$-reservation (as described above) has small expected maximum density, where the expectation is over the random choice of the $2$-reservation. 
It is easy to bound $\E[\text{density of }R]$ for any $R\sse[m]$.  We use Chernoff bounds to show that the density of $R$ is ``close'' to its expectation, and then take a union bound over all sets $R\sse[m]$ to obtain a bound on $\E[\max_{R\sse[m]}(\text{density of $R$})]$. 
Although this entails taking a union bound over exponentially many node-sets, two key insights make the calculations go through: 
(1)~if $|R|$ is small, then the expected density of $R$ is small, and so a tail event corresponds to a large relative deviation from the expectation; 
(2)~if $|R|$ is large, then the expected density of $R$ is large; however, the tail event corresponds to a large absolute deviation from the expectation. 
In either case, applying Chernoff bounds yields a rapid decay in probability, which turns out to be strong enough to survive the union bound.  

\subsection{Price of Reservation Bounds for Related Machines}

Our result in \Cref{relporlb-thm} is based on the intuition that any $k$-reservation can at best average out the load of $k$ speed classes. 
So, the hard instance (shown in \Cref{table:porlb}) is chosen such that the set of machines $M$ can be partitioned into $k+1$ speed classes $C_1,\ldots,C_{k+1}$ such that: 
(a)~the speed $\s_{\ell}$ of class $\ell$ decreases by some factor $\scaleparm$ over the classes, where $\scaleparm \approx \sqrt[2k]{|M|}$; 
and 
(b)~the total processing power (i.e., $|C_\ell| \cdot \s_{\ell}$) within each class increases by a factor $\scaleparm$ over the classes. 
The job-set $\jobset$ consists of $n \gg |M|$ i.i.d. jobs. 
The (common) job-size distribution is such that the random variable $N_\ell$ denoting the number of jobs that realize a size equal to the speed of class $\ell$ has expected value $|C_\ell|$. 
The rest of the probability mass is located at size $0$. 
The instance is explicitly stated in \Cref{table:porlb} in \Cref{sec:related-por}. 

By concentration, we can show that with high probability $N_\ell = O(|C_\ell|)$ holds for all $\ell \in [k+1]$. 
Thus, the omniscient optimum is $O(1)$. 
However, in any $k$-reservation, some $n/(k+1)$ jobs (call this set $J^\prime$) are devoid of a class-$\ell$ machine for some $\ell \in [k+1]$. 
This implies $\OPT_k = \Omega(\scaleparm/k)$ because neither the faster speed classes nor the slower speed classes are equipped to handle the load that arises from jobs in $J^\prime$ in a typical realization. 
Faster classes lack sufficient aggregate processing power, whereas slower classes are too slow to handle a large job-size realization without the load blowing up to $\Omega(\scaleparm)$. 

\subsection{(Bicriteria) Approximation Algorithms for Related Machines}

The situation for related machines is much more technically challenging. 
In particular, trivial extensions of our ideas from the identical-machines setting can lead to very poor approximation guarantees. 
For instance, suppose we have $m+1$ related machines, where the first machine has speed $\sqrt{m}$ and the remaining $m$ machines have unit speed. 
Suppose further that there are $n \gg m$ i.i.d. jobs, where the size of each job is distributed as $m/n + \sqrt{m} \cdot \mathrm{Bernoulli}(1/n)$. 
Observe that an $\Omega(\sqrt{m})$ loss in approximation is unavoidable in either of the following cases: 
(i)~the $2$-reservation for each job intersects at most one speed class; or (ii)~the two machines are sampled proportionally to their speeds, possibly from different speed classes. 

Our approximation algorithms for related machines are based on linear programming techniques. 
We begin with two preprocessing steps.  
First, we ``smooth'' the instance into a more convenient form. 
Among other things, the machines are partitioned into classes $C_1,\ldots,C_L$, where the speed $\s_{\ell}$ of class-$\ell$ machines decreases geometrically with $\ell$, while the processing power $|C_\ell| \cdot \s_{\ell}$ increases geometrically with $\ell$. 
This is inspired by the smoothing operation in the work of Im, Kell, Panigrahi and Shadloo~\cite{IKPS18}. 
It can be shown that this smoothing changes $\OPT_k$ by at most a constant factor. 
Second, via binary search, we suitably scale the instance so that $\OPT_k = \Omega(1)$. 
The point of this scaling is to formulate a linear program (LP) for the scaled instance such that infeasibility of the LP certifies that $\OPT_k$ is significantly larger than $1$. 
To simplify the exposition, we assume throughout that $\OPT_k = 1$. 
We also restrict attention to the case $k=2$, since the main ideas already appear in full generality in this special setting. 

A natural linear relaxation allows jobs to fractionally reserve up to $k$ machines.
We relax this modeling choice even further by allowing a job to reserve up to $k$ speed classes. 
Specifically, for each job $j \in \jobset$ and each pair $1 \leq \ell_1 < \ell_2 \leq L$, we introduce a variable $x_{j,\ell_1,\ell_2} \in [0,1]$, where $x_{j,\ell_1,\ell_2} = 1$ models that job $j$ reserves the speed classes $C_{\ell_1}$ and $C_{\ell_2}$. 
For each $j \in \jobset$, we impose the (assignment) constraint $\sum_{1 \leq \ell_1 < \ell_2 \leq L} x_{j, \ell_1, \ell_2} = 1$, which models a relaxed first-stage 2-class reservation. 

We next add volume constraints to control the load induced on the speed classes in the second stage. 
Since we assumed $\OPT_k = 1$, it is useful, for each pair $1 \leq \ell_1 < \ell_2 \leq L$, to partition the realizations of $P_j$, relative to the thresholds $\s_{\ell_2}$ and $\s_{\ell_1}$, into three regimes: small, medium, and large.  
Thus, define $P_{j,\ell_1,\ell_2}^{\sml} := P_{j} \cdot \bone(P_j < \s_{\ell_2})$, $P_{j,\ell_1,\ell_2}^{\med} := P_{j} \cdot \bone(\s_{\ell_2} \leq P_{j} < \s_{\ell_1})$, and $P_{j,\ell_1,\ell_2}^{\lrg} := P_{j} \cdot \bone(P_{j} \geq \s_{\ell_1})$. 
Here, $P^\sml_{j,\ell_1,\ell_2}$ represents realizations that can be processed even by a slow machine in $C_{\ell_2}$ without exceeding load $1$. 
The random variable $P^{\med}_{j,\ell_1,\ell_2}$ represents realizations that can be handled by a machine in $C_{\ell_1}$, but not by one in $C_{\ell_2}$, without exceeding load $1$. 
Finally, $P^{\lrg}_{j,\ell_1,\ell_2}$ captures realizations that cause load at least $1$ regardless of whether the job is assigned to a machine in $C_{\ell_1}$ or $C_{\ell_2}$. 

Since we assumed $\OPT_k = 1$, for each $\ell \in [L]$ we include the volume constraint: 
\begin{equation} \label{eq:truncvol}
\sum_{j \in \jobset, 1 \leq \ell^\prime < \ell} \E[P_{j,\ell^\prime,\ell}^{\sml}] \cdot x_{j,\ell^\prime,\ell} + \sum_{j \in \jobset, \ell < \ell^\prime \leq L} \E[P_{j,\ell,\ell^\prime}^{\med}] \cdot x_{j,\ell,\ell^\prime} \leq |C_\ell| \cdot \s_\ell.
\end{equation}
For $\ell^\prime < \ell$, the small part $P_{j,\ell^\prime,\ell}^{\sml}$ could also be processed on the faster machines in $C_{\ell^\prime}$ without exceeding load $1$. 
However, at the cost of an $O(1)$ loss in the makespan objective, we may charge these small realizations to the slower speed class $C_\ell$. 
This is because, after smoothing, the total processing power in class $\ell$ dominates that of all faster classes combined. 

Finally, the large realizations $\{ P_{j,\ell_1,\ell_2}^{\lrg} \}$ behave qualitatively differently from the small and medium parts: if the event $\{P_{j,\ell_1,\ell_2}^{\lrg} > 0\}$ happens, then assigning $j$ to any machine in $C_{\ell_1} \cup C_{\ell_2}$ already incurs load at least $1$. 
Thus, these correspond to exceptional overflow events rather than class-specific volume. 
It therefore suffices to control their total expected contribution via a single aggregate constraint. 
As in \cite{GuptaKNS21,IbrahimpurS20}, we include the following exceptional-jobs constraint:
\begin{equation} \label{eq:excepvol}
    \sum_{j \in \jobset} \sum_{1 \leq \ell_1 < \ell_2 \leq L} \Bigl(\E[P_{j,\ell_1,\ell_2}^{\lrg}]/\s_{\ell_1}\Bigr) \cdot x_{j,\ell_1,\ell_2} \leq 1.
\end{equation}

Now let $x^* \in [0,1]^{\jobset \times \binom{L}{2}}$ be a fractional solution to this feasibility LP.  
Observe that the coefficients in \eqref{eq:truncvol} are at most $\s_\ell$, while those in \eqref{eq:excepvol} are at most $1$. 
Since $|C_\ell|$ grows geometrically with $\ell$, this (assignment) LP satisfies a bounded-column-sums property. 
An iterative rounding result of Linhares, Olver, Swamy and Zenklusen \cite{LinharesOSZ20} therefore yields an integral $2$-class reservation $\psi$ that approximately preserves the volume guarantees satisfied by $x^*$. 

At this point, however, we still have reservations to two speed classes rather than to two specific machines. 
This is where we must give up something.  
A natural approach is to reserve, for each job $j$, one uniformly random machine from each of the two speed classes in $\psi(j)$. 
Due to the usual balls-and-bins behavior, this can incur an $O(\log m / \log \log m)$ loss in approximation. 
If we instead allow the reservation size to double, then we can define a $4$-reservation by choosing \emph{two} (uniformly) random machines from each of the two speed classes in $\psi(j)$. 
Each speed class can then be analyzed similarly to the identical-machines setting, yielding only a constant-factor loss relative to $\OPT_k$. 

\subsection{Effectiveness of  \texorpdfstring{\boldmath $2$}{2}-Reservations versus Adaptive Optimum for Related Machines}
 
Our key technical insight for proving \Cref{reladap-thm} is that the adaptive optimum is lower bounded, up to constant factors, by the ``cost'' of an optimal $1$-class reservation. 
Here, a $1$-class reservation means that in the first stage, for each job, we must commit to (or reserve) a single speed class that will be responsible for processing whatever size this job realizes in the second stage. 
(The second-stage assignment can choose any one machine from the reserved speed class for each job.) 
Thus, by sampling two machines independently and uniformly from the class that is reserved for each job, we can average out the load within each speed class. 
The associated $2$-reservation has an expected cost that is roughly the adaptive optimum.

\section{Preliminaries} \label{sec:prelims}

This section collects some preliminary tools that will be used in the rest of the paper. 
The following lemma gives bounds on the optimal makespan of an assignment that is \consistent with a reservation by relating it to a certain density function over subsets of machines. 

\subsection{Density Lower Bounds}

\begin{lemma} \label{lem:density}
Let $\jobset$ be a set of jobs, $M$ be a set of related machines with positive speeds $\{s_i\}_{i \in M}$, and $\sg : \jobset \to \subsets{M}$ be a reservation. 
For any arbitrary job sizes $p \in \Rp^{\jobset}$, define: 
\vspace{-5pt}
\begin{equation} \label{def:density}
\denlb(p;\sg,s) \coloneqq \max_{R \subseteq M : R \neq \emptyset} \frac{\sum_{j \in \jobset: \sg(j) \subseteq R} p_j}{\sum_{i \in R} s_i}.
\end{equation}
The following are true:
\begin{enumerate}[(i)]
    \item For any $\sg$-\consistent assignment $\asgn : \jobset \to M$, we have $\mksp(\asgn;p) \geq \denlb(p;\sg,s)$. 

    \item There is a $\sg$-\consistent assignment $\asgn : \jobset \to M$ satisfying 
    \[ 
    \mksp(\asgn;p) \leq \denlb(p;\sg,s) \cdot \max_{j \in \jobset} |\sg(j)|.
    \]

    \item Suppose there is a nontrivial partition of $M$ into nonempty subsets $C_1,\ldots,C_L$ such that for every job $j \in \jobset$ and every index $\ell \in [L]$, either $\sg(j) \subseteq C_\ell$ or $\sg(j) \cap C_\ell = \emptyset$ holds. 
    Then, the maximum in \eqref{def:density} is attained by a nonempty $R \subseteq C_\ell$ for some $\ell \in [L]$. 
\end{enumerate}
\end{lemma}

\begin{proof}
We establish the first part by showing that $\denlb(p;\sg,s)$ equals the optimal makespan that is achievable when the jobs are divisible, i.e., jobs can be run on one or more reserved machines, simultaneously, as long as the total portions add up to $1$. 
We say that a fractional assignment $x \in [0,1]^{M \times \jobset}$ is $\sg$-\consistent if for all $j \in \jobset$, we have $\sum_{i \in M} x_{ij} = 1$ and $x_{ij} = 0$ for any $i \notin \sg(j)$. 
For a fractional $\sg$-\consistent assignment $x$, we define $\mksp(x;\sg,s) = \max_i (\sum_j p_j x_{ij})/s_i$. 
Clearly, the makespan of any $\sg$-\consistent assignment $\asgn$ is bounded below by the optimal makespan achievable by a fractional $\sg$-\consistent assignment. 

We now interpret fractional assignments as flows in a suitable digraph $D = (V,A,b : V \to \R, u: A \to \Rp)$ whose nodes have demands $b_v$ and arcs have nonnegative capacities $u_e$. 
The makespan of $x$ will follow from a certain feasibility condition.  
The node-set $V$ comprises a node for each job $j \in \jobset$ with $b_j = -p_j$, a node for each machine $i \in M$ with $b_i = 0$, and a special sink-node $t$ with $b_t = \sum_j p_j$. 
The arc-set $A$ comprises an arc $(j,i)$ for every job $j \in \jobset$ and every reserved machine $i \in \sg(j)$, with $u_{j,i} = \infty$. It also comprises an arc $(i,t)$ for every machine $i \in M$, with $u_{i,t} = \ld s_i$ for some parameter $\ld \in \Rp$ to be defined shortly. 
We choose $\ld$ to be the smallest real number that guarantees the existence of a feasible flow satisfying the (node) demands and the (arc) capacities. 
Note $\ld \geq p(\jobset)/s(M)$ holds. 
Note that there is a one-to-one correspondence between feasible flows and optimal fractional assignments that are \consistent with $\sg$, simply by equating $x_{ij}$ and the portion of job $j$'s demand that is routed through the arc $(j,i)$. 
The cut conditions imply that the following constraints are necessary and sufficient for the existence of a feasible flow:
\[
\forall \text{ nonempty } X \subseteq \jobset \cup M, \, b(\jobset \cap X)  + u(\delout(X)) \geq 0 \text{ holds}, 
\]
where $\delout(X)$ is the set of arcs going from $X$ to $V \sm X$. 
Since $u_{j,i} = \infty$ whenever $i \in \sg(j)$, we only need to check the cut condition for those $X$ satisfying $j \in X \cap \jobset \implies \sg(j) \subseteq X \cap M$. 
Moreover, for any fixed machine set $R \subseteq M$, the most stringent finite cut with machine side $R$ is obtained by including precisely those jobs $j$ with $\sg(j) \subseteq R$.
In other words, the following condition is necessary and sufficient for the existence of a feasible flow: for any nonempty $R \sse M$, we have $\ld \cdot \sum_{i \in R} s_i \geq \sum_{j \in \jobset : \sg(j) \subseteq R} p_j$. 
This is precisely the form of the density lower bound, so $\ld = \denlb(p;\sg,s)$ holds. 

The second part is straightforward. 
Let $x$ be an optimal fractional assignment that is $\sg$-\consistent. 
By feasibility of $x$, for every job $j$, there is a machine $i(j) \in \sg(j)$ with $x_{ij} \geq 1/|\sg(j)|$. 
We define $\asgn$ by setting $\asgn(j) = i(j)$ for all $j$, i.e., we round up one of the large variables for $j$.
As $\mksp(x;\sg,s) = \ld = \denlb(p;\sg,s)$, we get that for any machine $i \in M$, $\sum_j p_j x_{ij} \leq \ld s_i$ holds. 
Therefore, 
\[
\mksp(\asgn;p) = \max_i \frac{\sum_{j : \asgn(j) = i} p_j}{s_i} \leq \max_i \frac{\sum_{j} (p_j \cdot x_{ij} \cdot |\sg(j)|) }{s_i} \leq \ld \cdot \max_{j \in \jobset} |\sg(j)|.
\]

Last, suppose $M = C_1 \sqcup \dots \sqcup C_L$ admits a nontrivial partition such that for every job $j$, $\sg(j)$ is contained in one of the $C_\ell$s. 
Suppose $R$ is a nonempty set that is not contained in any of the $C_\ell$s and is a maximizer in the definition of $\denlb(p;\sg,s)$. 
Define $J_0 \coloneqq \{ j : \sg(j) \subseteq R\}$, and for any $\ell \in [L]$, $J_\ell \coloneqq \{ j : \sg(j) \subseteq R \cap C_\ell \}$. 
By our assumption on the reservation $\sg$, we have $J_0 = \bigcup_{\ell} J_\ell$. 
By definition, $J_\ell$s are disjoint and $s(R) = \sum_{\ell \in [L] : R \cap C_\ell \neq \emptyset} s(R \cap C_\ell)$. 
Clearly,
\[
\denlb(p;\sg,s) = p(J_0)/s(R) \leq \max_{\ell \in [L] : R \cap C_\ell \neq \emptyset} p(J_\ell)/s(R \cap C_\ell).
\]
Therefore, some nonempty $(R \cap C_\ell)$-set is also a maximizer in the definition of $\denlb(p;\sg,s)$. 
\end{proof}

\subsection{Smoothing Related Machines and Associated Results}

Throughout most of our algorithms for the related-machines setting, we use the following notion of \smoothed machines, which was previously considered by Im et al.~\cite{IKPS18}. 

\begin{definition}\label{def:smoothedmachines} 
We say that a set $M$ of related machines is \emph{\smoothed} if it can be partitioned into speed classes $C_1,\ldots,C_L$ with the following properties:
\begin{enumerate}[(P1)]
\item (Rounded Speeds) Each machine in class $C_\ell$, $\ell \in [L]$, has a common speed $\s_\ell$, and this speed is a power of $2$. \label{prop:smoothpowersof2} 
    \item (Geometric Speeds) We have $\s_{\ell} \leq \s_{\ell-1}/2$ for any $\ell \in \{2,\ldots,L\}$. 
    \label{prop:smoothgeomspeeds}
    \item (Geometric Processing Power) We have $|C_\ell| \cdot \s_\ell \geq \sum_{\ell^\prime = 1}^{\ell-1} |C_{\ell^\prime}| \cdot \s_{\ell^\prime}$ for any $\ell \in [L]$. \label{prop:smoothgeomcomputes}
\end{enumerate}
\end{definition}

\begin{remark} \label{remark:smoothedmachines} 
A \smoothed set of machines $M = \bigcup_{\ell=1}^{L} C_\ell$ has $|C_\ell| \geq 2^{\ell-1}$ for all $\ell \in [L]$, and hence $L \leq \log_2 |M|$. 
\end{remark}

The following technical lemma will be useful in showing that \smoothing the machine speeds does not change the optimal value by more than a constant factor for the \slbr instances that we consider. 

\begin{lemma} \label{lem:smoothing}
For any given set $M$ of $m$ related machines, we can efficiently compute a \smoothed collection $M^\prime$ of $O(m^2)$ related machines, described by the speed class partitioning $C_1,\ldots,C_L$, such that: 
\begin{enumerate}[(i)]
    \item (Forward Approximation Loss) There is a universal mapping $f : M \to M^\prime$ such that for any $p \in \Rp^{\jobset}$ and any $\asgn : \jobset \to M$ we have $\mksp(f \circ \asgn;p) \leq 12 \cdot \mksp(\asgn;p)$. 
    \label{prop:smoothinglossforward}
    \item (Reverse Approximation Loss) There is a universal mapping $g : M^\prime \to M$ such that for any $p \in \Rp^{\jobset}$ and any $\asgn : \jobset \to M^\prime$ we have $\mksp(g \circ \asgn;p) \leq \mksp(\asgn;p)$. \label{prop:smoothinglossreverse} 
\end{enumerate}
\end{lemma}

We remark that we have not tried to optimize the forward approximation loss. 
The following notion of splitting a machine into one or more slower machines will be useful to us in \smoothing the machine speeds. 

\begin{definition} \label{def:machinesplitting}
A finite set $M$ of one or more related machines with speeds $\{s_r\}_{r \in M}$ is said to be a \emph{splitting} of a machine with speed $s \in \R_{>0}$  if $s_r > 0$ for all $r \in M$ and $\sum_{r \in M} s_r = s$. 
\end{definition}

Note that the trivial splitting corresponds to just retaining the machine as is. 
The following lemma shows that the splitting operation does not lead to a decrease in the makespan objective. 
This will be relevant in establishing property \ref{prop:smoothinglossreverse}. 
The proof follows from the definition of $q$-norms; the makespan objective is the $\infty$-norm of the load vector. 

\begin{lemma} \label{lem:splittingislossy}
Suppose that we have a set of $m$ related machines with positive speeds $s_1,\ldots,s_m$, and a prescribed splitting $(M_i, \{s_r\}_{r \in M_i})$ of machine $i$ for each $i \in [m]$. 
Let $p \in \R_{\geq 0}^{\jobset}$ be an arbitrary job-size vector and $\asgn : \jobset \to \bigcup_{i \in [m]} M_i$ be an arbitrary assignment of the jobs to the split machines.   
Define $\asgn^\prime : \jobset \to [m]$ as follows: for each job $j \in \jobset$, $\asgn^\prime(j)$ is the unique $i$ such that $\asgn(j) \in M_i$, i.e., $\asgn^\prime$ is obtained by undoing the splitting operation. 
For any $q \geq 1$, we have $\norm{\loadvec(\asgn^\prime,p)}_q \leq \norm{\loadvec(\asgn,p)}_q$. 
\end{lemma}
\begin{proof}
Fix a machine $i \in [m]$ and consider the set of jobs $J_i \coloneqq \{ j \in \jobset : \asgn(j) \in M_i\}$ that are processed by a machine in $M_i$ under the assignment $\asgn$. 
Using an averaging argument, we have: 
\begin{multline*}
\load(i;\asgn^\prime,p) = \frac{p(J_i)}{s_i} \leq \max_{r \in M_i} \frac{\sum_{j : \asgn(j) = r} p_j}{s_r} = \max_{r} \load(r;\asgn,p) \leq \norm{(\load(r;\asgn,p))_{r \in M_i}}_q.
\end{multline*}
We raise the first and the last terms in the above inequality to their $q$th powers to obtain: 
\begin{equation*}
\bigl(\norm{\loadvec(\asgn^\prime,p)}_q\bigr)^q = \sum_{i \in [m]} \bigl(\load(i;\asgn^\prime,p)\bigr)^q \leq \sum_{i \in [m], r \in M_i} \bigl(\load(r;\asgn,p)\bigr)^q = \bigl(\norm{\loadvec(\asgn,p)}_q\bigr)^q. \tag*{\qedhere} 
\end{equation*}
\end{proof}

Observe that it is straightforward to ensure that the \smoothed machines satisfy \Cref{lem:smoothing}~\ref{prop:smoothinglossreverse} because none of the following operations leads to a decrease in the makespan objective: rounding down the speeds; splitting a machine into one or more slower machines; and discarding some machines.  
Thus, the crux is in performing the \smoothing operation while ensuring  \Cref{lem:smoothing}~\ref{prop:smoothinglossforward}. 

\begin{proofof}{\Cref{lem:smoothing}}
We prove the lemma by describing the procedure for obtaining the set of \smoothed machines $M^\prime$ through a sequence of \smoothing steps where we keep track of the mappings $f : M \to M^\prime$ and $g : M^\prime \to M$.
To minimize notational heft, we assume that the multiplicative gap between the fastest and slowest speed is at most $m$ and that the speeds are already \emph{rounded}.
For the former, we may discard every machine whose speed is smaller than $s_{\max}/m$, where $s_{\max}$ is the largest machine speed, and map such machines to a fastest machine. 
Doing this may increase the makespan objective by at most a factor $2$ since the total speed of all discarded machines is at most the speed of a fastest machine.   
For the latter, we round down all the speeds to the nearest power of $2$, and this can lead to an increase in the makespan by at most a factor $2$. 
Overall, the forward approximation loss, attributable to the above \smoothing steps, is at most $4$, so there is still room for another loss of a factor $3$ to establish property~\ref{prop:smoothinglossforward}. 
Note that the makespan objective does not decrease when speeds are rounded down. 
The mapping $g : M^\prime \to M$ will be implicitly defined as each machine $r \in M^\prime$ arises from some machine $i \in M$ via splitting, so $g(r) \coloneqq i$. 
Thus, by Lemma~\ref{lem:splittingislossy}, claim~\ref{prop:smoothinglossreverse} (reverse approximation loss) always holds. 
The mapping $f : M \to M^\prime$, however, will require some care in its definition, so we update it in each \smoothing step. 
We remark that $f(i)$ for a machine $i \in M$ is updated at most once, and this is done just before $i$ undergoes a nontrivial splitting. 

Initially, $M^\prime$ comprises a copy of each machine in $M$, and both $f : M \to M^\prime$ and $g : M^\prime \to M$ are equal to the corresponding identity maps. 
We use $i$ and $r$ to index a machine in $M$ and $M^\prime$, respectively. 
If the index $i$ is used to denote a machine in $M^\prime$, then it necessarily means that the machine has not undergone any nontrivial splittings, and so it appears in both $M$ and $M^\prime$. 
We group the machines in $M^\prime$ based on their speeds, and let $C_1,\ldots,C_L$ be the speed classes in decreasing order of speeds. 
Clearly, the roundedness property~\ref{prop:smoothpowersof2} and the geometric-speeds property~\ref{prop:smoothgeomspeeds} are satisfied. 
We shall ensure that these properties are maintained at the end of each of the upcoming \smoothing steps. 
Next, we focus on attaining the geometric-processing-powers property~\ref{prop:smoothgeomcomputes} iteratively for each $\ell \in [L]$ in increasing order of $\ell$. 
The property is vacuously true for $\ell = 1$. 
Suppose property~\ref{prop:smoothgeomcomputes} holds for all $\ell \in [\ell_0]$, for some $\ell_0 \in [L]$, and let $\ell_0$ denote the largest index for which this is true. 
If $\ell_0 = L$, then we have nothing left to do,  so we terminate the \smoothing procedure. 
Otherwise, consider the largest $\ell_1 \in \{\ell_0+1,\ldots,L\}$ for which the following inequality holds: 
\[
\sum_{\ell=\ell_0+1}^{\ell_1} |C_\ell| \cdot \s_\ell < \sum_{\ell = 1}^{\ell_0} |C_{\ell}| \cdot \s_{\ell}.
\]
We remark that such an $\ell_1$ exists by the maximality of $\ell_0$. 
Note that property~\ref{prop:smoothgeomcomputes} for $\ell_0$ implies that the RHS in the above inequality can be upper bounded by ${2 |C_{\ell_0}| \s_{\ell_0}}$. 
A nice consequence of this observation and the fact that the speeds are powers of $2$ is that the entire load on machines in $\bigcup_{\ell = \ell_0+1}^{\ell_1} C_{\ell}$ can be reassigned to (or redistributed among) the machines in $C_{\ell_0}$. 
This reassignment does not make the makespan  (over the machines in $\bigcup_{\ell = \ell_0}^{\ell_1} C_{\ell}$) grow by a factor greater than $3$, for any choice of job sizes $p$ and under any assignment $\asgn$. 
We use this reassignment to update the mapping $f$: for any machine $i \in \bigcup_{\ell = \ell_0+1}^{\ell_1} C_{\ell}$, we set $f(i) \coloneqq r$ for a suitable machine $r \in C_{\ell_0}$. 
Note here that $\bigcup_{\ell = \ell_0+1}^{\ell_1} C_{\ell} \subseteq M$ since none of these machines have undergone a nontrivial splitting so far. 
However, as will be clear soon, machines in $C_{\ell_0}$ may have arisen from nontrivial splittings of machines in $M$ that were performed in previous iterations.   
Formally, we choose the reassignment so that the following inequality holds for all $r \in C_{\ell_0}$; such a reassignment exists by a greedy packing argument, since the total speed to be reassigned is less than $2 |C_{\ell_0}| \s_{\ell_0}$ and every reassigned machine has speed at most $\s_{\ell_0}/2$:
\begin{equation} \label{eq:smoothreassignment}
\sum_{\ell=\ell_0+1}^{\ell_1} \s_\ell \cdot |\{ i \in C_\ell : f(i) = r \} | \leq 2 \cdot \s_{\ell_0}.
\end{equation}
If $\ell_1 = L$ at this point, then we discard all machines in $\bigcup_{\ell = \ell_0+1}^{\ell_1} C_{\ell}$ from $M^\prime$ and terminate the \smoothing procedure; this is done because there are not enough machines left to form another speed class with sufficiently large total processing power. 
Otherwise, for each $\ell \in \{\ell_0+1,\ldots,\ell_1\}$ and each machine $i \in C_\ell$, we split the machine $i$ into $\s_\ell / \s_{\ell_1+1}$ slower machines each with speed $\s_{\ell_1+1}$; we implicitly set $g(r)$ to $i$ for each of the split machines $r$ and remove $i$ from $C_\ell$. 
Each of the classes $C_{\ell_0+1},\ldots,C_{\ell_1}$ becomes empty after the above splittings have been performed. 
We discard these empty classes, re-index the subsequent classes, and update $L$. 
Observe that property~\ref{prop:smoothgeomcomputes} now holds for $\ell_0 + 1$. 
We repeat the above procedure iteratively until $\ell_0$ or $\ell_1$ equals $L$. 
Since the ratio between the largest and smallest speed is at most $m$ after the initial reduction, each original machine is split into at most $m$ machines throughout the procedure. 
Thus, $|M^\prime| = O(m^2)$.

It is straightforward to see that properties~\ref{prop:smoothpowersof2}-\ref{prop:smoothgeomcomputes} of a \smoothed collection hold at the end of the \smoothing procedure.  
We already argued that claim~\ref{prop:smoothinglossreverse} follows from Lemma~\ref{lem:splittingislossy}. 
Claim~\ref{prop:smoothinglossforward} holds because of our choice of updating $f$ to satisfy \eqref{eq:smoothreassignment}. 
\end{proofof}

Where necessary to disambiguate the instance $\I$ under consideration when describing the cost of a reservation $\mset$, we use the notation $\cost_\I(\mset)$.
This is especially relevant when the machine speeds are modified. 

\begin{lemma} \label{lem:optundersmoothing} 
Suppose we are given an \slbr instance $\I = (\{P_j\}_{j \in \jobset}, \{s_i\}_{i \in M}, k)$ where $\jobset$ is the job-set, $M$ is a set of related machines, and $P_j$ is the size-distribution of job $j$. 
We can efficiently compute a set of \smoothed machines $M^\prime$ with $|M^\prime| = O(|M|^2)$ such that the instance $\I^\prime$ obtained from $\I$ by replacing $M$ with $M^\prime$ has the following properties: 
\begin{enumerate}[(a)]
    \item For any reservation $\sg : \jobset \to \subsets{M}$, there is a reservation $\sgp : \jobset \to \subsets{M^\prime}$ with $|\sgp(j)| \leq |\sg(j)|$ for all $j \in \jobset$ that also satisfies $\cost_{\I^\prime}(\sgp) = O(\cost_{\I}(\sg))$. 

    \item For any reservation $\sgp : \jobset \to \subsets{M^\prime}$, there is a reservation $\sg : \jobset \to \subsets{M}$ with $|\sg(j)| \leq |\sgp(j)|$ for all $j \in \jobset$ that also satisfies $\cost_{\I}(\sg) \leq \cost_{\I^\prime}(\sgp)$. 
\end{enumerate}
In particular, for all $k \in \Z_{>0}$, we have $\OPT_k(\I) \leq \OPT_k(\I^\prime) \leq O(\OPT_k(\I))$.
\end{lemma}

\begin{proof}
    This is now an immediate corollary of \Cref{lem:smoothing}.
    Construct the \smoothed  collection $M^\prime$ according to \Cref{lem:smoothing}, and let $\I^\prime$ denote the corresponding instance.
        Given a reservation $\sg$ for $\I$, simply take $\sgp(j)$ to be the image of $\sg(j)$ under $f$. Note that $|\sgp(j)| \leq |\sg(j)|$. 
    Then for any $\sg$-\consistent assignment $\asgn$, the assignment $\asgn^\prime = f \circ \asgn$ is $\sgp$-\consistent in $\I^\prime$, and $\mksp_{\I^\prime}(\asgn^\prime;P) \leq 12 \cdot \mksp_{\I}(\asgn;P)$. (Here, we subscript by the appropriate instance to ensure no ambiguity.) 
    Conversely, given a reservation $\sgp$ for $\I^\prime$, take $\sg(j)$ to be the image of $\sgp(j)$ under $g$. Note that $|\sg(j)| \leq |\sgp(j)|$.
    Then for any $\sgp$-\consistent assignment $\asgn^\prime$, the assignment $\asgn = g \circ \asgn^\prime$ is $\sg$-\consistent in $\I$, and $\mksp_{\I}(\asgn;P) \leq \mksp_{\I^\prime}(\asgn^\prime;P)$. 
\end{proof}

\subsection{Concentration Bounds}

We now state Chernoff/Hoeffding concentration inequalities that are frequently used in our results. 
The general form is given below. 
For a reference, see the chapter entitled \emph{Chernoff and Hoeffding Bounds} in \cite{MU-book}. 

\begin{theorem} \label{thm:chernoff} 
Let $X_1,\dots,X_n$ be independent $[0,1]$-bounded random variables, and $X = \sum_{i=1}^n X_i$ denote their sum.
We have: 
\begin{enumerate}
\item For any $\mu \geq \E[X]$ and $\delta \geq 0$, we have:
\begin{equation} \label{eq:chernoffuppertail}
\Pr[X \geq (1+\delta)\mu] \leq \left(\frac{\euler^{\delta}}{(1+\delta)^{1+\delta}}\right)^{\mu} 
\leq \left(\frac{\euler}{1+\delta}\right)^{(1+\delta)\mu} 
\end{equation}
\item For any $\mu \leq \E[X]$ and $\delta \in [0,1)$, we have:
\begin{equation} \label{eq:chernofflowertail}
    \Pr[X \leq (1-\delta)\mu] \leq \left(\frac{\euler^{-\delta}}{(1-\delta)^{1-\delta}}\right)^{\mu} \leq \euler^{-\mu \delta^2/2}.
\end{equation}
\end{enumerate}

\end{theorem}

\noindent The following result is derived by integrating the upper-tail bound. 
This result is convenient when $\delta = \Omega(1)$.  

\begin{lemma} \label{lem:uppertailmass}
Let $X_1,\dots,X_n$ be independent $[0,1]$-bounded random variables and $X \coloneqq \sum_{i=1}^n X_i$ denote their sum. 
For any $\delta \geq \euler^2 - 1$ and $B \geq (1+\delta) \E[X]$, we have $\E[\max(0,X-B)] \leq (\euler/(1+\delta))^B$. 
\end{lemma}

\begin{proof}
For any $t \geq 0$ and $\mu \coloneqq (B+t)/(1+\delta) \geq \E[X]$, the upper-tail bound in Theorem~\ref{thm:chernoff} yields $\Pr[X \geq B+t] \leq (\euler/(1+\delta))^{B+t}$. 
Thus, 
\begin{equation*}
\E[\max(0,X-B)] = \int_0^{\infty} \Pr[X \geq B+t] dt \leq \left(\frac{\euler}{1+\delta}\right)^{B} \cdot \int_0^{\infty} \left(\frac{\euler}{1+\delta}\right)^{t} dt \leq \left(\frac{\euler}{1+\delta}\right)^{B},
\end{equation*}
where we use $(1+\delta)/\euler \geq \euler$ and $\int_0^{\infty} a^t dt = -1/\ln a$ when $a \in (0,1)$. 
\end{proof}

\subsection{Expected Maximum of Independent Random Variables}

The following lemma provides a convenient way to estimate the expected maximum of a collection of independent random variables.

\begin{lemma}[See \cite{IbrahimpurS20}] \label{lem:rho2apx}
Let $\{X_i\}_{i \in I}$ be a finite collection of independent nonnegative random variables with finite mean. 
Let $\theta \geq 0$ be a real number.  
The following are true: 
\begin{enumerate}[(i)]
\item If $\sum_i \E[X_i \cdot \bone(X_i \geq \theta)] > \theta$, then $\E[\max_i X_i] > \theta/2$.
\item If $\sum_i \E[X_i \cdot \bone(X_i \geq \theta)] \leq \theta$, then $\E[\max_i X_i] \leq 2 \theta$.
\end{enumerate}
\end{lemma}

\section{Identical Machines} \label{sec:identical} 

We prove Theorem~\ref{thm:iden2options} in this section.
We show that the stated guarantee holds in fact for a random $2$-reservation $\sg$, obtained as follows: for each job $j \in \jobset$, we sample two machines $i_1(j)$ and $i_2(j)$ independently and uniformly from $[m]$, and set $\sg(j) = \{i_1(j),i_2(j)\}$; note $|\sg(j)| = 1$ if the same machine is sampled twice. 
Note that the computation of $\sg$ does not depend on the underlying job-size distributions. Consequently, the guarantee in Theorem~\ref{thm:iden2options} holds even with arbitrarily correlated jobs. 
The following claim is straightforward. 

\begin{claim} \label{clm:jlandsinR}
For $j \in \jobset$ and nonempty $R \subseteq [m]$ we have $\Pr[\sg(j) \subseteq R] = |R|^2/m^2$. 
\end{claim}

We show $\E_{\sg}[\cost(\sg)] = O(\OPT_m)$ by showing the same inequality for any choice of $p \in \Rp^{\jobset}$: $\E_{\sg}[\cost(\sg;p)] = O(\max(\max_j p_j,\sum_j p_j/m))$ holds. 
That is, choosing two machines independently and uniformly at random for each job is sufficient to ensure a constant-factor approximation guarantee against the optimal makespan. 
This holds even for an adversarial choice of $p$, as long as the adversary is oblivious to the choice of $\sg$. 
For the special case of deterministic unit-size jobs, a similar result was obtained by Azar et al.~\cite{AzarBKU99}; see their Lemma 6.1.

\begin{theorem} \label{thm:twochoicessuffice}
Let $p \in \Rp^{\jobset}$ be arbitrary, and let
$p_\mx \coloneqq \max_{j \in \jobset} p_j$ and
$p_\total \coloneqq \sum_{j \in \jobset} p_j$.
Then 
\[
\E_{\sg}[\cost(\sg;p)] = O(\max(p_\mx,p_\total/m)).
\]
\end{theorem}

\begin{proof}
For simplicity, we scale down the instance by $\max(p_\mx,p_\total/m)$ so that $p_j \in [0,1]$ for all $j$ and $p_\total \leq m$ holds. 
It remains to argue that $\E_{\sg}[\cost(\sg;p)] = O(1)$.
By Lemma~\ref{lem:density}, we have:
\begin{flalign} \label{eq:idencostubden}
\cost(\sg;p) & = \min_{\substack{\sg\text{-\consistent} \\ \text{assignment } \asgn}} \mksp(\asgn;p) \leq 2 \cdot \max_{R \subseteq [m] : R \neq \emptyset} \frac{\sum_{j : \sg(j) \subseteq R} p_j}{|R|}.
\end{flalign}
We take an expectation over inequality~\eqref{eq:idencostubden} to obtain: 
\begin{flalign}
& \E_{\sg}[\cost(\sg;p)] \leq 2 \cdot \E_{\sg} \Bigl[ \max_{R \subseteq [m] : R \neq \emptyset} \frac{\sum_{j} p_j \cdot \bone(\sg(j) \subseteq R)}{|R|} \Bigr] \\ 
& \leq 2 \Bigl\{ \euler^2  + \sum_{R \subseteq [m] : R \neq \emptyset} \frac{1}{|R|} \E_{\sg} \Bigl[ \max\bigl(0, \sum_{j} p_j \cdot \bone(\sg(j) \subseteq R) - \euler^2 |R| \bigr) \Bigr] \Bigr\}, \label{eq:idenunionbound1}
\end{flalign}
where we use a union bound in the last inequality. 
Now, observe that for any nonempty $R \subseteq [m]$ the random variable $X \coloneqq \sum_{j} p_j \cdot \bone(\sg(j) \subseteq R)$ is a sum of independent $[0,1]$-bounded random variables. 
By Claim~\ref{clm:jlandsinR}, $\E[X] = p_{\total} \cdot |R|^2/m^2 \leq |R|^2/m$. 
Using Lemma~\ref{lem:uppertailmass} with $\delta = \euler^2 m/|R| - 1$ and $B = \euler^2 |R|$, we get: 
\[
\E_{\sg}[ \max(0, X - \euler^2 |R|)]  \leq (\euler m / |R|)^{-\euler^2 |R|}.
\]
Plugging the above bound into \eqref{eq:idenunionbound1} and using the well-known bound $\binom{m}{\ell} \leq (\euler m/\ell)^{\ell}$ we get:
\[
\E_{\sg}[\cost(\sg;p)] \leq 2 \euler^2 + 2 \cdot \sum_{\ell=1}^m \binom{m}{\ell} \cdot \frac{1}{\ell} \cdot \Bigl( \frac{\euler m}{\ell} \Bigr)^{-\euler^2 \ell} \leq 2 \euler^2 + 2 \cdot \sum_{\ell=1}^m \frac{\euler^{-6\ell}}{\ell} \leq 15. \tag*{\qedhere}
\]
\end{proof}

Our main result follows as a corollary of the above theorem. 

\begin{proof}[\Cref{thm:iden2options}]
The omniscient optimum is at least $\max(P_\mx,P_\total/m)$, where $P_{\mx} \coloneqq \max_j P_j$ and $P_{\total} \coloneqq \sum_j P_j$ are random variables that denote the maximum and the sum, respectively, of the job random variables.  
Thus, 
\begin{multline*} 
\E_{\sg}[\cost(\sg)] = \E_{\sg}[\E_P[\cost(\sg;P)]] = \E_{P}[\E_{\sg}[\cost(\sg;P)]] \\ \leq \E_{P}[15 \cdot \max(P_{\mx},P_{\total}/m)] = O(\OPT_m), 
\end{multline*}
where we use the Fubini-Tonelli theorem to exchange the expectations in the second equality.  
\end{proof}

\section{Price of Reservation Bounds for Related Machines} \label{sec:related-por}

In this section, we derive upper and lower bounds on the price of $k$-reservation for related machines.  
All our results in this section assume that the $P_j$s are mutually independent. 

\subsection{Lower Bounds} \label{relpor-lb}

Our main result of this section is the following. 

\relporlbthm* 

As a corollary, we can infer that the price of $k$-reservation for related machines is unbounded, if $k=o(\log m/\log \log m)$.

\begin{corollary} 
For any constant $c$, there is an \slbr instance on $m$ related machines, where $m$ is sufficiently large, such that $\OPT_k / \OPT_m \geq c$ for $k=o(\log m/\log \log m)$. 
\end{corollary}

We start with a description of the instance $\I = (\jobset,M,(P_j)_{j \in \jobset})$ that is used in establishing Theorem~\ref{relporlb-thm}. 
Fix some integers $k \geq 1$ and $m \geq 2 (2k+2)^{2k}$.  
Let $\scaleparm$ be the largest power of $2$ such that $\sum_{\ell=0}^k \scaleparm^{2\ell} \leq m$ holds. 
Note that $\scaleparm \geq k+1$.
By setting some machine speeds to $0$, if necessary, we may assume that $\sum_{\ell=0}^k \scaleparm^{2\ell} = m$ holds. 
The machine-set $M$ has $k+1$ speed classes $C_1,\ldots,C_{k+1}$, where for any $\ell \in [k+1]$ the class $C_\ell$ has $\scaleparm^{2\ell-2}$ machines each with speed $\s_\ell \coloneqq \scaleparm^{k-\ell+1}$. 
Note $|C_\ell| \s_\ell = \scaleparm^{k+\ell-1}$ for any $\ell$. 
The instance has $n \gg |M|$ identical jobs. 
Every $P_j$ is an independent and identically distributed copy of a common random variable $X$.
This random variable $X$ is supported on the $(k+2)$-point set $\{0, 1, \scaleparm, \ldots,\scaleparm^k\}$ and has a probability mass of $\scaleparm^{2\ell-2}/n$ at the point $\scaleparm^{k+1-\ell}$ for each $\ell \in [k+1]$. 
The remaining probability mass of $1 - (\sum_{\ell=0}^{k} \scaleparm^{2\ell})/n$ is located at $0$.  
\Cref{table:porlb} summarizes important statistics of the instance. 

\begin{table}[ht]
\scriptsize 
\centering
\begin{equation*}
{\renewcommand{\arraystretch}{1.5}
\begin{array}{|c|c|c|c|c|c|c|c|c|} 
\hline
\text{(class index)} \, \ell & 1 & 2 & \ldots & \ell-1 & \ell & \ell+1 & \ldots & k+1 \\
\hline
\text{(\# class-} \ell 
\text{ machines)} \, |C_\ell| & 1 & \scaleparm^2 & \ldots & \scaleparm^{2\ell-4} & \scaleparm^{2\ell-2} & \scaleparm^{2\ell} & \ldots & \scaleparm^{2k} \\
\hline
\text{(speed of class } \ell) \, \s_\ell & \scaleparm^k & \scaleparm^{k-1} & \ldots & \scaleparm^{k-\ell+2} & \scaleparm^{k-\ell+1} & \scaleparm^{k-\ell} & \ldots & 1 \\
\hline
\text{(processing power in class } \ell) \, |C_\ell| \s_\ell & \scaleparm^{k} & \scaleparm^{k+1} & \ldots & \scaleparm^{k+\ell-2} & \scaleparm^{k+\ell-1} & \scaleparm^{k+\ell} & \ldots & \scaleparm^{2k} \\
\hline 
\E[N_\ell] = \sum_{j \in \jobset} \Pr[P_j = \s_\ell] = n \Pr[X = \s_\ell] & 1 & \scaleparm^{2} & \ldots & \scaleparm^{2\ell-4} &  \scaleparm^{2\ell-2} & \scaleparm^{2\ell} & \ldots &  \scaleparm^{2k} \\
\hline 
\end{array}}
\end{equation*}
\caption{Instance $\I$ showing unbounded gap between $\OPT_k$ and $\OPT_m$.}
\label{table:porlb}
\end{table}

\begin{proof}[\Cref{relporlb-thm}]
We prove the theorem in two parts: we first show that the omniscient optimum is $O(1)$ using Chernoff bounds, and then we show that any $k$-reservation must have a cost of $\Omega(\scaleparm/k)$.

For the first part, we define an integer random variable $N_\ell \coloneqq \sum_{j \in \jobset} \bone(P_j = \s_\ell)$ that counts the number of jobs that realize a size of $\s_\ell$, for any index $\ell \in [k+1]$. 
We have $\E[N_\ell] = |C_\ell| = \scaleparm^{2\ell-2}$ (see \Cref{table:porlb}). 
For any real number $\alpha \geq 1$, consider the event $G(\alpha) \coloneqq \{ \, N_\ell \leq \alpha |C_\ell| \; \forall \ell \, \}$.  
Clearly, the omniscient optimum is at most 
$\ceil{\alpha}$ whenever the event $G(\alpha)$ happens: for any $\ell$, we assign up to $\ceil{\alpha}$ jobs of size $\s_\ell$ to each class-$\ell$ machine. 
So it suffices to show that $\E[\alpha(P)] = O(1)$, where $\alpha(P) \coloneqq \max_{\ell} N_\ell/|C_\ell|$. 
By a union-bound argument: 
\[
\E[\alpha(P)] \leq \euler^2 + \sum_{\ell=1}^{k+1} \E[\max(0,N_\ell - \euler^2 |C_\ell|)] \leq \euler^2 + \sum_{\ell=1}^{k+1} \euler^{-|C_\ell|} = O(1), 
\]
where we use \Cref{lem:uppertailmass} with $\delta = \euler^2 - 1$ and $B = \euler^2 |C_\ell|$ in the second inequality. 

For the second part, fix a $k$-reservation $\sg : \jobset \to \subsets{M}$. 
Since we have $k+1$ speed classes, there is an index $\ell \in [k+1]$ such that at least $n/(k+1)$ of the jobs are devoid of a reservation for a class-$\ell$ machine. 
Let $J^\prime \coloneqq \{ j \in \jobset : \sg(j) \cap C_\ell = \emptyset\}$ denote these jobs, and let $Z \coloneqq \sum_{j \in J^\prime} \bone(P_j = \s_{\ell})$ denote the random count of size $\s_{\ell}$ realizations among these jobs. 
Note $Z$ is a binomial random variable with mean $\E[Z] \geq \scaleparm^{2\ell-2}/(k+1)$. 
We consider two cases depending on whether $\ell$ equals $1$. 
When $\ell = 1$, there is an $\Omega(1/k)$ probability that some job in $J^\prime$ realizes size $\s_1$ (=$\scaleparm^k$). 
In this case, any $\sg$-\consistent assignment $\asgn$ that processes this $\scaleparm^k$-sized job on a machine from $C_{2} \cup \ldots \cup C_{k+1}$ must necessarily incur a makespan of at least $\scaleparm$. 
On the other hand, when $\ell > 1$, the bad event $B = \{ Z \geq \scaleparm^{2\ell-2}/(k+1) \}$ happens with constant probability. 
Suppose $B$ happens. 
Note that the total realized volume of jobs with size $\s_\ell (=\scaleparm^{k-\ell+1})$ is at least $\scaleparm^{k+\ell-1}/(k+1)$. 
Any $\sg$-\consistent assignment $\asgn$ that processes even a single job of size $\scaleparm^{k-\ell+1}$ on a machine from $C_{\ell+1} \cup \ldots \cup C_{k+1}$ must necessarily incur a makespan of at least $\scaleparm$. 
If this is not the case, a volume of at least $\scaleparm^{k+\ell-1}/(k+1)$ must be handled by faster classes $C_1 \cup \ldots \cup C_{\ell-1}$ which in total have a processing power of at most $2 \scaleparm^{k+\ell-2}$. 
In any case, the cost of $\sg$ is $\Omega(\scaleparm/k)$. 
Since $\sg$ was assumed to be arbitrary, the same bound holds for $\OPT_k$. 
\end{proof}

\subsection{Upper Bounds} \label{relpor-ub}

Our main result of this section is the following. 

\relporubthm*

\begin{proof}
We prove the theorem for $k = 4 \log m$. 
By monotonicity, the result follows for all larger $k$. 
Consider an arbitrary instance of $\slbr$ on related machines. 
We assume that the machine-set $M$ in the instance is \smooth (\Cref{def:smoothedmachines}), and let $C_1,\ldots,C_L$ denote its speed classes. 
By Lemma~\ref{lem:optundersmoothing}, the \smoothing step may lead to an $O(1)$-loss in the cost and may increase the number of machines quadratically. 
So, it suffices to establish $\OPT_{2L}/\OPT_{m} = O(1)$ for the modified instance, where we use that $L \leq \log_2 |M| \leq 2 \log m$. 

Consider the trivial reservation $\sgs$ that reserves all of $M$ for every job: $\sgs(j) = M$ for all $j$. 
Note that the omniscient optimum is always \consistent with $\sgs$, so $\cost(\sgs) = \OPT_m$. 
Let $t \coloneqq 80 \cdot \OPT_m$. 
We scale down all the $P_j$s by $t$ so that $\cost(\sgs) = 1/80$ holds. 
Consider the class reservation $\psi$ that reserves every speed class for every job: $\psi(j) = [L]$ for all $j$. 
Note that $\psi$ is induced by $\sgs$.  
We invoke Lemma~\ref{lem:multiassignmentcostlb} w.r.t. the reservation $\sgs$ to obtain the following: 
\begin{itemize}
    \item A volume bound on scaled super-exceptional portions: $\sum_{j \in \jobset} \bigl(\E[P_j \cdot \bone(P_j \geq \s_1)]/\s_1\bigr) \leq 1/40$. 
    \item A volume bound on the $\ell$-truncated portions for any $\ell \in [L]$: 
    \[
    \sum_{j \in \jobset} \E[P_j \cdot \bone(\s_{\ell+1} \leq P_j < \s_\ell)] \leq \sum_{j \in \jobset} \E[P_j \cdot \bone(\s_{\ell+1} \leq P_j)] \leq |C_\ell| \s_\ell,
    \]
    where we recall $\s_{L+1} \coloneqq 0$. 
\end{itemize}

Consider the following random $2L$-reservation $\sg$: for each job $j \in \jobset$ and each index $\ell \in [L]$, we sample two machines independently and uniformly from class $C_\ell$, and reserve the sampled machines for $j$. 
Note that $\psi$ is also the induced class reservation for $\sg$. 
The above volume bounds imply that the hypothesis of Lemma~\ref{lem:twochoicesperclass} (stated and proved in \Cref{related-bicriteria}) holds for $\beta = 1$, and thus, $\E_{\sg}[\cost(\sg)] = O(1)$ holds.
We are done as we had scaled down the $P_j$s so that $\OPT_m = 1/80$ holds. 
\end{proof}

The above proof strategy implies a slightly stronger guarantee for \smoothed instances that have a small number of speed classes.  

\begin{theorem} 
Consider an instance $\I$ of $\slbr$ where the set of machines $M = C_1 \cup \ldots \cup C_L$ forms a \smoothed collection and admits a partition into $L$ speed classes. 
The price of $2L$-reservation for $\I$ is $O(1)$.  
\end{theorem}

\section{Approximation Algorithms for \boldmath \slbr on Related Machines}
\label{sec:related}

In this section, we provide (bicriteria) approximation algorithms for the related-machines setting. 
We assume that the given set $M$ of $m$ related machines forms a \smoothed collection and that the job sizes $\{P_j\}_{j \in \jobset}$ are mutually independent.  
By Lemma~\ref{lem:optundersmoothing}, restricting attention to \smoothed machines changes the optimal value by at most a constant factor. 
Recall that a \smoothed collection $M$ can be partitioned into geometric speed classes $C_1,\ldots,C_L$, where $L = O(\log m)$, such that: (i)~for every $\ell \in [L]$, all machines in $C_\ell$ have the same speed $\s_\ell$, and the speeds $\s_1,\ldots,\s_L$ are distinct powers of $2$; and (ii)~for any $\ell > 1$, $|C_\ell| \s_\ell \geq \sum_{\ell^\prime < \ell} |C_{\ell^\prime}| \s_{\ell^\prime}$. 

\subsection{Lower Bound on the Cost of a Reservation} 

In this section, we derive sufficient conditions for a reservation $\sg$ to have a large cost. 
This lower bound (Lemma~\ref{lem:multiassignmentcostlb}) will inspire the valid inequalities that we include in the LP relaxation for the problem.  
We introduce some notation so that the cost of a reservation can be captured clearly and concisely. 
First, we extend the $\s_\ell$ notation for the speed of class $\ell$ by defining $\s_0 \coloneqq \infty$ and $\s_{L+1} \coloneqq 0$. 
Consider a nonempty index set $K \subseteq [L]$; later, $K$ will typically have the interpretation as the set of speed classes in the reservation of some job. 
Then for any $\ell \in \{0\} \cup [L]$, we define $\nxt(\ell;K) \coloneqq \min \{ \ell^\prime \in K \cup \{L+1\} : \ell^\prime > \ell\}$ as the smallest index in $K \cup \{L+1\}$ that exceeds $\ell$. 
Note that $\nxt(0;K) \in K$ is the index of the fastest speed class in $K$. 
For any job $j \in \jobset$ and any nonempty $K \subseteq [L]$, we define a decomposition of its random size $P_j$ as
\[ 
P_j = P_{j,K,0} + \sum_{\ell \in K} P_{j,K,\ell}, 
\]
where
\begin{equation} \label{eq:rangetruncatedrv} 
P_{j,K,\ell} \coloneqq P_j \cdot \bone(\s_{\nxt(\ell;K)} \leq P_j < \s_\ell) \qquad \text{ for each $\ell \in \{0\} \cup K$}.
\end{equation}
We refer to $P_{j,K,0}$ as the \emph{super-exceptional} or \emph{$0$-exceptional} portion, and $P_{j,K,\ell}$ as the \emph{$\ell$-truncated portion} whenever $\ell \in K$. 
We can also extend the notion of $\ell$-exceptional random variables to all $\ell \in [L]$ as follows: for any job $j$ and any nonempty $K \subseteq [L]$, 
\begin{equation} \label{eq:ellexceptionalrv} 
P_{j,K,\leq \ell} \coloneqq \sum_{\ell^\prime \in \{0\} \cup K : \ell^\prime \leq \ell} P_{j,K,\ell^\prime} = P_j \cdot \bone(P_j \geq \s_{\nxt(\ell;K)}).
\end{equation}
Note that $P_{j,K,\leq 0} = P_{j,K,0}$ and $P_{j,K,\leq L} = P_j$. 

The intuition for the above decomposition is quite simple. 
Whenever a super-exceptional random variable $P_{j,K,0}$ realizes a positive size, no machine in a speed class from $K$ can process this job $j$ without the makespan being at least $1$. 
The same conclusion holds for the $\ell$-exceptional random variable $P_{j,K,\leq \ell}$ for any $\ell \in [L]$ (or the $\ell$-truncated variable $P_{j,K,\ell}$ for any $\ell \in K$) when considering a slower machine in $\bigcup_{\ell^\prime \in K : \ell^\prime > \ell} C_{\ell^\prime}$. 
On the other hand, for any $\ell \in K$, a single $\ell$-truncated random variable $P_{j,K,\ell}$ can be processed on any machine from class $C_\ell$ (or faster) without inducing a load greater than $1$. 
We illustrate the notion of $\ell$-truncated and super-exceptional portions in \Cref{fig:truncated-random-variables}. 

\begin{figure}[!t]
    \centering
    \begin{tikzpicture}[x=1.25cm,y=1.05cm,font=\small]
        \def\xSix{1.00}
        \def\xFive{1.90}
        \def\xFour{2.80}
        \def\xThree{3.70}
        \def\xTwo{4.65}
        \def\xOne{5.60}
        \def\xInf{7.40}

        \def\densitypath{%
            (0,0.18)
            .. controls (0.45,1.65) and (0.82,2.75) .. (1.35,2.40)
            .. controls (1.95,2.02) and (2.30,1.12) .. (2.85,1.35)
            .. controls (3.35,1.60) and (3.95,2.05) .. (4.55,1.60)
            .. controls (5.10,1.18) and (5.85,0.65) .. (7.40,0.08)%
        }

        \begin{scope}
            \clip (0,0) rectangle (\xSix,3.2);
            \fill[gray!15]
                (0,0) -- \densitypath -- (\xInf,0) -- cycle;
        \end{scope}

        \begin{scope}
            \clip (\xSix,0) rectangle (\xThree,3.2);
            \fill[gray!30]
                (0,0) -- \densitypath -- (\xInf,0) -- cycle;
        \end{scope}

        \begin{scope}
            \clip (\xThree,0) rectangle (\xOne,3.2);
            \fill[gray!48]
                (0,0) -- \densitypath -- (\xInf,0) -- cycle;
        \end{scope}

        \begin{scope}
            \clip (\xOne,0) rectangle (\xInf,3.2);
            \fill[
                pattern=north east lines,
                pattern color=black
            ]
                (0,0) -- \densitypath -- (\xInf,0) -- cycle;
        \end{scope}

        \draw[->]
            (0,0) -- (7.75,0)
            node[right] {realized size $p_j$};

        \draw[->]
            (0,0) -- (0,3.15)
            node[above] {};

        \draw[very thick] \densitypath;

        \foreach \x/\lab in {
            \xSix/{\s_6},
            \xFive/{\s_5},
            \xFour/{\s_4},
            \xThree/{\s_3},
            \xTwo/{\s_2},
            \xOne/{\s_1}}
        {
            \draw[densely dashed,gray]
                (\x,0) -- (\x,3.00);

            \draw
                (\x,0.06) -- (\x,-0.06)
                node[below=2pt] {$\lab$};
        }

        \node[below=2pt] at (0,0) {$0$};
        \node[below=2pt] at (\xInf,0) {$\infty$};

        \draw[
            decorate,
            decoration={brace,mirror,amplitude=4pt}
        ]
            (0,-0.6) -- (\xSix,-0.6)
            node[midway,below=6pt]
            {$P_{j,\psi(j),6}$};

        \draw[
            decorate,
            decoration={brace,mirror,amplitude=4pt}
        ]
            (\xSix,-0.6) -- (\xThree,-0.6)
            node[midway,below=6pt]
            {$P_{j,\psi(j),3}$};

        \draw[
            decorate,
            decoration={brace,mirror,amplitude=4pt}
        ]
            (\xThree,-0.6) -- (\xOne,-0.6)
            node[midway,below=6pt]
            {$P_{j,\psi(j),1}$};

        \draw[
            decorate,
            decoration={brace,mirror,amplitude=4pt}
        ]
            (\xOne,-0.6) -- (\xInf,-0.6)
            node[midway,below=6pt]
            {$P_{j,\psi(j),0}$};
    \end{tikzpicture}

    \caption{Schematic illustration of the distribution of the job-size random variable $P_j$ in an instance with six speed classes. For the class reservation $\psi(j) = \{1,3,6\}$, the range of possible job sizes is partitioned into the $6$-truncated portion on $[0,\s_6)$, the $3$-truncated portion on $[\s_6,\s_3)$, the $1$-truncated portion on $[\s_3,\s_1)$, and the super-exceptional portion on $[\s_1,\infty)$. The thresholds $\s_5$, $\s_4$, and $\s_2$ correspond to speed classes not included in $\psi(j)$.}
    \label{fig:truncated-random-variables}
\end{figure}
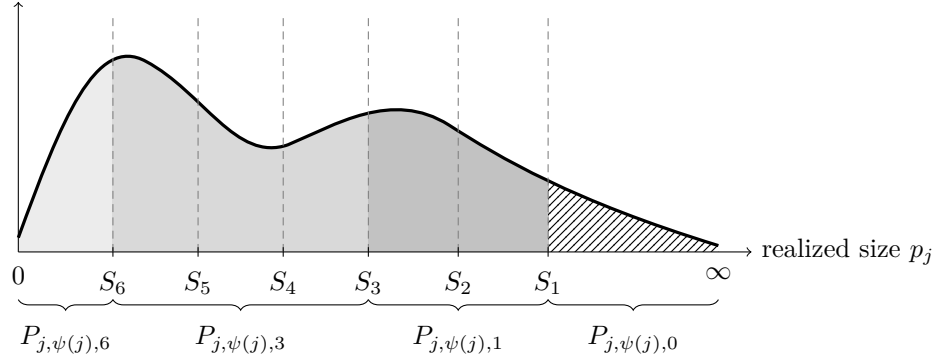

\begin{lemma} \label{lem:multiassignmentcostlb}
Let $\sg : \jobset \to 2^{[m]}$ be a reservation of jobs to machines, and let $\psi : \jobset \to 2^{[L]}$ be the induced class reservation given by $\psi(j) = \{ \ell \in [L] : \sg(j) \cap C_\ell \neq \emptyset \}$. 
Then the following are true:
\begin{enumerate}[(a)]
\item $\cost(\sg) \geq \frac12 \cdot \min\bigl\{1,\sum_{j \in \jobset} \bigl(\E[P_{j,\psi(j),0}] / \s_{\nxt(0;\psi(j))}\bigr) \bigr\}$. \label{lemitem:excepvolumelbzero}
\item  If $\sum_{j \in \jobset} \E[P_{j,\psi(j),\leq \ell}] >  \sum_{\ell^\prime=1}^{\ell} |C_{\ell^\prime}| \cdot \s_{\ell^\prime}$ for some $\ell \in [L]$, then $\cost(\sg) > 1/80$. \label{lemitem:excepvolumelbell}
\end{enumerate}
\end{lemma}

Modulo the super-exceptional portion, the above lemma gives conditions for a reservation to have a small cost: the total mass of $\ell$-exceptional random variables must be bounded by the total processing power in the $\ell$ fastest speed classes for every $\ell$. 
A nice feature of the lemma is that it provides a means to estimate $\cost(\sg)$ via a scaling argument. 
Note that the terms appearing in Lemma~\ref{lem:multiassignmentcostlb} can be efficiently computed given some basic information about the job random variables. 
Since $\E[\cost(\sg;P)] = t \cdot \E[\cost(\sg;P/t)]$ for any real $t > 0$, a simple binary search can be used to estimate $\cost(\sg) = \E[\cost(\sg;P)]$ within constant factors. 
Another feature of our lemma is that it only uses terms such as $\E[P_{j,\psi(j),\leq \ell}]$ that are nonincreasing in the scaling parameter $t$. 

\smallskip

The following claim will be useful to us in the proof of Lemma~\ref{lem:multiassignmentcostlb} as it describes a lower bound on the second-stage makespan in terms of the realized $\ell$-exceptional jobs. 

\begin{claim} \label{clm:multiassignmentcostdet}
Let $\sg$ and $\psi$ be as in the statement of Lemma~\ref{lem:multiassignmentcostlb}. 
Fix some $\ell \in \{0\} \cup [L]$ and let $p \in \Rp^{\jobset}$ be such that $p_j > 0 \implies p_j \geq \s_{\nxt(\ell;\psi(j))}$ for all $j$. 
For any $\sg$-\consistent assignment $\asgn$, we have: 
\begin{enumerate}[(i)]
    \item If $\ell = 0$, then $\mksp(\asgn;p) \geq \max_{j} p_j/ \s_{\nxt(\ell;\psi(j))}$.
    \item If $\ell = 1$, then $\mksp(\asgn;p) \geq \min\{1,(\sum_j p_j)/(|C_1|\s_1)\}$.
    \item If $\ell > 1$, then $\mksp(\asgn;p) \geq \min\{1,(\sum_j p_j)/(2|C_\ell|\s_\ell)\}$. 
\end{enumerate}
\end{claim}

\begin{proof}
The first part is trivial as the makespan is at least the load induced by any single job. 
Suppose $\ell \geq 1$. 
For the makespan to be less than $1$, every job $j$ with positive size must be processed on some machine in $C_1 \cup \ldots \cup C_\ell$, as the $p_j$s are $\ell$-exceptional. 
By property~\ref{prop:smoothgeomcomputes} of \smoothed machines, the total processing power in the $\ell$ fastest speed classes is at most $|C_1| \s_1$ when $\ell = 1$ and $2 \cdot |C_\ell| \cdot \s_\ell$ when $\ell > 1$. 
The desired conclusions follow from an averaging argument.  
\end{proof}

\begin{proofof}{Lemma~\ref{lem:multiassignmentcostlb}}
We start with the proof of claim (a).  
Define $Q_j \coloneqq P_{j,\psi(j),0} / \s_{\nxt(0,\psi(j))}$, and let $T \coloneqq \sum_{j \in \jobset} \E[Q_j]$ be the total mass of scaled super-exceptional variables.  
By definition of $\sg$-compatibility, $Q_j$ is a lower bound on the second-stage makespan regardless of which machine in $\sg(j)$ processes job $j$. 
Therefore, $\cost(\sg) \geq \E[\max_{j \in \jobset} Q_j]$. 
We prove the first part by showing that $\E[\max_j Q_j] \geq \min(1,T)/2$ holds. 
By definition, $\supp(Q_j) \subseteq \{0\} \cup [1,\infty)$. 
The contrapositive of Lemma~\ref{lem:rho2apx}~(i) for $\theta = 1$ states that $\sum_j \E[Q_j] > 1 \implies \E[\max_j Q_j] > 1/2$, so we are done when $T > 1$. 
Now suppose $T \leq 1$. 
Since $Q_j$ does not take any size in the open interval $(0,1)$, for any $\theta \in (0,T)$ we have: 
\[
\sum_j \E[Q_j \cdot \bone(Q_j \geq \theta)] = \sum_j \E[Q_j] = T > \theta.
\]
Thus, the contrapositive of Lemma~\ref{lem:rho2apx}~(i) implies $\E[\max_j Q_j] > \theta/2$ for every $\theta < T$. 
It follows that $\E[\max_j Q_j] \geq T/2$, which completes the proof of claim (a). 

\smallskip

We now prove claim (b). 
The heart of the proof is in showing that with $\Omega(1)$ probability, the total realized volume of the $\ell$-truncated random variables is $\Omega(|C_\ell| \cdot \s_\ell)$. 
Claim~\ref{clm:multiassignmentcostdet} would then imply $\cost(\sg;P) = \Omega(1)$. 

\smallskip

Let $\ell \in [L]$ be minimal such that $\sum_{j \in \jobset} \E[P_{j,\psi(j),\leq \ell}] >  \sum_{\ell^\prime=1}^{\ell} |C_{\ell^\prime}| \cdot \s_{\ell^\prime}$.
We consider the cases $\ell=1$ and $\ell > 1$ separately.

Suppose $\ell=1$.
By definition of $1$-exceptional variables (see \eqref{eq:ellexceptionalrv}), the volume inequality can be rewritten as: 
\[
\sum_{j \in \jobset} \E[P_{j,\psi(j),0}] + \sum_{j \in J^\prime} \E[P_{j,\psi(j),1}] > |C_1| \cdot \s_1,
\] 
where $J^\prime \coloneqq \{ j \in \jobset : 1 \in \psi(j)\}$. 
If the quantity $T$ from the first part of the proof is at least $1/2$, then we are already done because $\cost(\sg) \geq 1/4$ holds. 
Otherwise: 
\[
\sum_{j \in \jobset} \E[P_{j,\psi(j),0}] \leq \s_1 \cdot \sum_{j \in \jobset} \bigl(\E[P_{j,\psi(j),0}]/\s_{\nxt(0;\psi(j))}\bigr) \leq \s_1/2 \leq |C_1|\s_1/2.
\]
Combining the above two displayed inequalities, we get $\sum_{j \in J^\prime} \E[P_{j,\psi(j),1}] > |C_1| \cdot \s_1/2$. 
Now let $R_j \coloneqq P_{j,\psi(j),1} / \s_1$ for $j \in J^\prime$ and $X \coloneqq \sum_{j \in J^\prime} R_j$ denote their sum. 
Note that the $R_j$s form an independent collection of $[0,1]$-bounded random variables, and $\E[X] > |C_1|/2$. 
Using the lower tail bound in Theorem~\ref{thm:chernoff} for $\mu = |C_1|/2$ and $\delta = 1/2$, we get $\Pr[X \leq |C_1|/4] \leq \exp(-|C_1|/16) < 0.95$. 
Therefore, the event $\{\sum_{j \in J^\prime} P_{j,\psi(j),1} > |C_1|\s_1/4 \}$ happens with probability greater than $1/20$. 
By Claim~\ref{clm:multiassignmentcostdet}, we get $\cost(\sg) > 1/80$.

Now suppose $\ell > 1$. So 
\[ \sum_{j \in \jobset} \E[P_{j,\psi(j),\leq \ell}] > \sum_{\ell^\prime=1}^{\ell} |C_{\ell^\prime}|\s_{\ell^\prime}, \] and by minimality of $\ell$, 
\[ \sum_{j \in \jobset} \E[P_{j,\psi(j),\leq \ell-1}] \leq \sum_{\ell^\prime=1}^{\ell-1} |C_{\ell^\prime}|\s_{\ell^\prime}. \] 
Combining these two inequalities, we get $\sum_{j \in J^\prime} \E[P_{j,\psi(j),\ell}] > |C_\ell| \cdot \s_\ell$, where $J^\prime \coloneqq \{ j \in \jobset : \ell \in \psi(j)\}$. 
As before, let $R_j \coloneqq P_{j,\psi(j),\ell} / \s_\ell$ for $j \in J^\prime$ and $X \coloneqq \sum_{j \in J^\prime} R_j$ denote their sum. 
Note that the $R_j$s form an independent collection of $[0,1]$-bounded random variables, and $\E[X] > |C_\ell|$. 
Using the lower tail bound in Theorem~\ref{thm:chernoff} for $\mu = |C_\ell|$ and $\delta = 1/2$, we get $\Pr[X \leq |C_\ell|/2] \leq \exp(-|C_\ell|/8) < 0.9$. 
Therefore, the event $\{\sum_{j \in J^\prime} P_{j,\psi(j),\ell} > |C_\ell|\s_\ell/2 \}$ happens with probability greater than $1/10$. 
By Claim~\ref{clm:multiassignmentcostdet}, we get $\cost(\sg) > 1/40$. 
\end{proofof}

\subsection{\boldmath Linear Relaxation for \texorpdfstring{$\slbr$}{\slbr} on Related Machines} \label{subsec:lpoptk}

Fix an optimal first-stage reservation $\sgs$ for the given instance of \slbr.   
We now write a compact feasibility LP whose infeasibility implies $\cost(\sgs) = \Omega(1)$. 
Let $\indexset \coloneqq \{ K \subseteq [L] : 1 \leq |K| \leq k\}$.
For each job $j \in \jobset$ and each index set $K \in \indexset$, we have a Boolean decision variable $x_{j,K}$ modeling $K = \{ \ell \in [L] : \sgs(j) \cap C_\ell \neq \emptyset\}$. 
These decision variables are relaxed to take values in the interval $[0,1]$ to obtain a linear program.  
Note that the number of choices for $K$ is polynomial, even if $k$ is not constant, since $L$ has logarithmic size (Remark~\ref{remark:smoothedmachines}).
We use $\F$ to denote the feasible region of the program below.   

\begin{align} 
& \sum_{j \in \jobset, K \in \indexset} (\E[P_{j,K,0}]/\s_{\nxt(0;K)}) \cdot x_{j,K} \, \, \leq \, \, 1 &  \label{eq:superexceptionalvolume} \\
& \sum_{j \in \jobset, K \in \indexset} \E[P_{j,K,\leq \ell}] \cdot x_{j,K} \, \, \leq \, \,  |C_1| \cdot \s_1 +\ldots + |C_\ell| \cdot \s_\ell & \forall \ell \in [L]  \label{eq:ellexceptionalvolume} \\
& \sum_{K \in \indexset} x_{j,K} \, \,  = \, \, 1 &  \forall j \in \jobset \label{eq:reservationconstraint} \\ 
& x_{j,K} \, \, \geq \, \,  0 & \forall j \in \jobset, K \in \indexset \label{eq:nonnegativity} 
\end{align}

We also define $\F(t)$ to be this same polytope, but for processing times that are all scaled down by a factor of $t$ (so $\F(1) = \F$).
A straightforward consequence of Lemma~\ref{lem:multiassignmentcostlb} is the following. 

\begin{lemma} \label{feasibleguess}
For any $t \geq 80 \cdot \cost(\sgs)$, we have $\F(t) \neq \emptyset$.
\end{lemma}
\begin{proof}
By homogeneity of $\cost$ it suffices to establish the claim for $t = 1$ under the assumption that $\cost(\sgs) \leq 1/80$. 
Let $\psi : \jobset \to \indexset$ denote the class reservation induced by $\sgs$, i.e., $\psi(j) = \{\ell \in [L] : \sgs(j) \cap C_\ell \neq \emptyset\}$ for all $j$. 
Note that $1 \leq |\psi(j)| \leq |\sgs(j)| \leq k$ for all $j$. 
Consider the $\{0,1\}$-solution $x \in [0,1]^{\jobset \times \indexset}$ defined as $x_{j,K} = 1$ if and only if $K = \psi(j)$. 
We show that $x \in \F$, thereby proving the claim. 
Nonnegativity and reservation constraints are satisfied by the definition of $x$. 
Since $x$ is a $\{0,1\}$-solution, we have:  
\[
\sum_{j \in \jobset, K \in \indexset} (\E[P_{j,K,0}]/\s_{\nxt(0;K)}) \cdot x_{j,K} = \sum_{j \in \jobset} (\E[P_{j,\psi(j),0}]/\s_{\nxt(0;\psi(j))}) \leq 1/40
\]
and 
\[
\sum_{j \in \jobset, K \in \indexset} \E[P_{j,K,\leq \ell}] \cdot x_{j,K} = \sum_{j \in \jobset} \E[P_{j,\psi(j),\leq \ell}] \leq |C_1| \cdot \s_1 +\ldots + |C_\ell| \cdot \s_\ell
\]
for all $\ell \in [L]$, where the first upper bound follows from   Lemma~\ref{lem:multiassignmentcostlb}~\ref{lemitem:excepvolumelbzero} and the second upper bound follows from the contrapositive of Lemma~\ref{lem:multiassignmentcostlb}~\ref{lemitem:excepvolumelbell}.  
\end{proof}

Let $t^* > 0$ be the infimum over real numbers such that $\F(t^*) \neq \emptyset$. 
By \Cref{feasibleguess}, $t^* \leq 80 \cdot \cost(\sgs)$, so $t^*/80$ serves as a lower bound on the optimal value for the given instance of \slbr.  
The following claim shows that $t^*$ lies in a real interval whose endpoints differ by a factor of $O(m)$. 
Therefore, a binary search can be used to obtain $t \approx t^*$ using $O(\log m)$ invocations of an algorithm that tests whether $\F(t)$ is empty. 

\begin{claim}
Define $\lb \coloneqq \bigl(\sum_{j \in \jobset} \E[P_j] \bigr)/\s_1$. 
We have $\lb/m \leq t^* \leq 80 \cdot \lb$. 
\end{claim}

\begin{proof}
Consider the trivial reservation $\sg$ that assigns every job to one fixed machine in $C_1$. 
Clearly, $\cost(\sg) = (\sum_j \E[P_j])/\s_1$, so we get $t^* \leq 80 \cdot \cost(\sgs) \leq 80 \cdot \lb$. 
Next, by definition of $\ell$-exceptional variables \eqref{eq:ellexceptionalrv}, $P_{j,K,\leq L} = P_j$ for any job $j$ and index-set $K \subseteq [L]$. 
Hence, the LHS of the volume inequality \eqref{eq:ellexceptionalvolume} for $\ell = L$ equals $\sum_j \E[P_j]$, where we use $\sum_{K} x_{j,K} = 1$ for all $j$, while the RHS is at most $m \s_1$. 
Therefore, $\F(t) = \emptyset$ whenever $t < \lb/m$, and hence $t^* \geq \lb/m$.  
\end{proof}

\subsection{Rounding a Fractional Class Reservation}

By scaling if necessary, we can assume for the remainder that $\F \neq \emptyset$ and $\OPT_k = \Omega(1)$, i.e., that $t^* \approx 1$.
We now describe an iterative rounding algorithm that takes as input a fractional class reservation $x \in \F$ and outputs an integral class reservation $\psi$ that approximately satisfies the volume constraints in \eqref{eq:ellexceptionalvolume}. 
We use the following iterative rounding result of Linhares et al.~\cite{LinharesOSZ20} in a black-box fashion. 

\begin{theorem}[Follows from Corollary 11 in~\cite{LinharesOSZ20}] 
\label{thm:iterativerounding} 
Let $\M=(\U,\I)$ be a matroid with rank function
$\rank$ (specified by a value oracle), 
and let $\Q \coloneqq \{z\in\Rp^{\U}:
z(\U)=\rank(\U),\ \ z(R)\leq\rank(R)\ \forall R \subseteq \U\}$ be its base polytope. 
Let $\bz$ be a feasible solution to the following multi-budgeted matroid LP: 
\begin{equation*} 
\min \ \ c^\top z \qquad \text{s.t.} \qquad 
A z \leq b, \quad z \in \Q.
\end{equation*}
where $A \in \Rp^{k \times \U}$. 
Suppose that for any $e \in \supp(\bz)$ we have $\sum_{i \in [k]} A_{i,e} \leq \viol$
for some parameter $\viol$ and the rank function $\rank(R)$ can be evaluated in $\poly(|\U|)$ time for any $R \subseteq \U$. 
Then we can round $\bz$ in polynomial time to obtain a basis $B$ of $\M$ 
satisfying: 
(i)~$B \sse \supp(\bz)$;
(ii)~$c(B) \leq c^\top \bz$; and 
(iii)~$A\chi^B \leq b + \viol \mathbf{1}$, where $\chi^B$ is the characteristic vector of $B$ and $\mathbf{1}$ is the vector of all
$1$s.
\end{theorem}

\begin{lemma} \label{lem:relatedcruderounding}
For any given $x \in \F$, we can efficiently compute an integral class reservation $\psi : \jobset \to \indexset$ satisfying:
\begin{enumerate}[(i)]
\item For every job $j \in \jobset$, $x_{j,\psi(j)} > 0$.
\item We have $\sum_{j \in \jobset} \bigl(\E[P_{j,\psi(j),0}] / \s_{\nxt(0;\psi(j))}\bigr) \leq \sum_{j \in \jobset, K \in \indexset} (\E[P_{j,K,0}]/\s_{\nxt(0;K)}) \cdot x_{j,K}$. 
\item  For any $\ell \in [L]$, we have $\sum_{j \in \jobset : \psi(j) \ni \ell} \E[P_{j,\psi(j),\ell}] \leq 4 \cdot |C_{\ell}| \cdot \s_{\ell}$. 
\end{enumerate}    
\end{lemma}
\begin{proof}
We prove this result using Theorem~\ref{thm:iterativerounding}. 
To apply this theorem, we need to set up a minimization LP for which the given point $x$ is feasible and the feasible region of this LP should be a matroid base polytope intersected with some packing constraints that have the so-called ``bounded column sums'' property. 

Let $\U \coloneqq \{ (j,K) : j \in \jobset, K \in \indexset \}$. 
Let $\Q$ be the assignment polytope for the job-set $\jobset$ and the machine-set $\indexset$:
\[
\Q \coloneqq \{ z \in \Rp^{\U} : 
\sum_{K \in \indexset} z_{j,K} = 1 \; \forall j \in \jobset \}
\]
It is well known that the assignment polytope is integral and it can be described as the base polytope of a partition matroid on the groundset $\U$. 
The extreme points of $\Q$ are in one-to-one correspondence with the set of all class reservations. 
We now define $L$ packing constraints of the form $Az \leq b$, where $A \in [0,1]^{[L] \times \U}$ and $b \in \Rp^{[L]}$.   
For every $(j,K) \in \U$ and $\ell \in K$, we set $A_{\ell,j,K} = \E[P_{j,K,\ell}]/(|C_\ell|\s_\ell)$ to be the expected value of the $\ell$-truncated variable for job $j$ w.r.t. the index-set $K$, scaled down by the total processing power in speed class $\ell$. 
By definition of $\ell$-truncated variables (see \eqref{eq:rangetruncatedrv}), we have $0 \leq A_{\ell,j,K} \leq 1/|C_\ell|$. 
For $\ell \notin K$, we set $A_{\ell,j,K} = 0$. 
Next, we set $b_\ell = 2$ for all $\ell \in [L]$.  
Last, for all $(j,K) \in \U$ we set $c_{j,K} = \E[P_{j,K,0}] / \s_{\nxt(0;K)}$ to be the expected mass of the super-exceptional variable for job $j$ w.r.t. index-set $K$, scaled down by the speed of the fastest machine that can process it. 

\smallskip
Consider the linear program: $\min\{c^\top z : Az \leq b, z \in Q\}$. 
The given $x$ is feasible to this LP because: (i)~$x \in \F \implies x \in \Q$; (ii)~for any $\ell \in [L]$, the volume inequality \eqref{eq:ellexceptionalvolume} and property~\ref{prop:smoothgeomcomputes} of \smoothed machines imply:  
\begin{multline*}
\sum_{j,K : K \ni \ell} A_{\ell,j,K} \cdot x_{j,K} = \frac{\sum_{j,K : K \ni \ell} \E[P_{j,K,\ell}] \cdot x_{j,K}}{|C_\ell| \cdot \s_\ell} \leq \frac{\sum_{j,K} \E[P_{j,K,\leq \ell}] \cdot x_{j,K}}{|C_\ell| \cdot \s_\ell} \leq \frac{\sum_{\ell^\prime=1}^{\ell} |C_{\ell^\prime}| \cdot \s_{\ell^\prime}}{|C_{\ell}| \cdot \s_\ell} \leq 2.
\end{multline*}
Next, we compute the maximum column sum in $A$. 
Fix some $(j,K)\in \U$. 
Recall that $A_{\ell,j,K} = 0$ for $\ell \notin K$ and $0 \leq A_{\ell,j,K} \leq 1/|C_\ell|$ for $\ell \in K$. 
Since the cardinality of the speed classes grows geometrically (\Cref{remark:smoothedmachines}), we get $\sum_{\ell \in [L]} A_{\ell,j,K} \leq \sum_{\ell \in K} \frac{1}{|C_\ell|} \leq 2$. 

\smallskip

We now invoke Theorem~\ref{thm:iterativerounding} for the above LP with $\viol = 2$ to obtain an integral solution $\hat{x} \in \Q$ satisfying: (i)~$\supp(\hat{x}) \subseteq \supp(x)$; (ii)~$c^\top \hat{x} \leq c^\top x$; and~(iii)~$A\hat{x} \leq 4 \cdot \bone$. 
This $\hat{x}$ corresponds to an integral class reservation $\psi : \jobset \to \indexset$ with the claimed properties.  
\end{proof}

\subsection{\boldmath An \texorpdfstring{$O(\log m/\log \log m)$}{O(log m/log log m)}-Approximation Algorithm}
\label{related-log}
We now describe a simple greedy strategy for converting a class reservation $\psi$ to a reservation $\sg$ where for each job $j$ we choose exactly one machine from $C_\ell$ for every class-index $\ell$ that appears in $\psi(j)$ (Lemma~\ref{lem:onechoiceperclass}). 
We also show that this step leads to at most a factor $O(\log m/\log \log m)$ loss in the cost (Lemma~\ref{lem:costofonechoice}). 

\begin{lemma} \label{lem:onechoiceperclass}
Let $\beta \geq 1$ be a real number and let $\psi$ be a class reservation such that for any class-index $\ell \in [L]$, we have $\sum_{j \in \jobset : \psi(j) \ni \ell} \E[P_{j,\psi(j),\ell}] \leq \beta \cdot |C_{\ell}| \cdot \s_{\ell}$. 
We can efficiently compute a reservation $\sg$ of jobs to machines satisfying: 
\begin{enumerate}[(i)]
    \item For every job $j \in \jobset$ and index $\ell \in \psi(j)$,  $\sg(j)$ contains a unique machine from $C_\ell$. 
    In other words, $|\sg(j)| = |\psi(j)|$ for all $j$, and $\psi$ is the class reservation induced by $\sg$. 
    
    \item For any index $\ell \in [L]$ and machine $i \in C_\ell$, we have $\sum_{j \in \jobset : \sg(j) \ni i} \E[P_{j,\psi(j),\ell}] \leq (\beta+1) s_i$.
\end{enumerate}
\end{lemma}
\begin{proof}
The reservation $\sg$ can be greedily constructed by iterating over the jobs and class-indices in $\psi(\cdot)$ in an arbitrary order. 
Initially, $\sg : \jobset \to \{\emptyset\}$ and all pairs $(j,\ell)$ with $\ell \in \psi(j)$ are unmarked. 
Consider some unmarked job-index pair $(j^*,\ell)$ s.t. $\ell \in \psi(j^*)$. 
Let $i \in C_\ell$ be the class-$\ell$ machine with the smallest value of $\sum_{j \in J : \sg(j) \ni i} \E[P_{j,\psi(j),\ell}]$ w.r.t. the partial reservation $\sg$ that has been constructed so far. 
As $P_{j,\psi(j),\ell}$ is always bounded from above by $s_i (= \s_\ell)$ for all relevant $j$ (recall \eqref{eq:rangetruncatedrv}), we get $\sum_{j \in J : \sg(j) \ni i} \E[P_{j,\psi(j),\ell}] \leq \beta \cdot s_i$ using an averaging argument and our assumption that the total $\ell$-truncated mass is at most $\beta \cdot |C_\ell| \cdot \s_\ell$. 
We include the machine $i$ in $\sg(j^*)$ (i.e., reserve $i$ for $j^*$) and mark the pair $(j^*,\ell)$. 
We repeat the above procedure as long as there is an unmarked pair. 
It is trivial to check that the resulting reservation $\sg$ satisfies all of the properties stated in the lemma.  
\end{proof}

Our next result is that the reservation $\sg$ produced by Lemma~\ref{lem:onechoiceperclass} is an $O(\log m/\log \log m)$-approximation to the given instance of \slbr. 
The proof is based on Chernoff bounds. 
The high-level idea is to exhibit a $\sg$-\consistent assignment $\asgn$, depending on the realization $P$, such that $\mksp(\asgn;P)$ has strong concentration properties. 
We recall \eqref{eq:rangetruncatedrv} where we gave a decomposition of a job variable $P_j$ into a super-exceptional portion and one or more $\ell$-truncated portions. 
We restate this decomposition w.r.t. the index-set $\psi(j)$:
\[
P_j = \sum_{\ell \in \{0\} \cup \psi(j)} P_{j,\psi(j),\ell}, 
\text{ where } P_{j,\psi(j),\ell} = P_j \cdot \bone(\s_{\nxt(\ell;\psi(j))} \leq P_j < \s_\ell).
\]

\begin{lemma} \label{lem:costofonechoice} 
Let $\beta \geq 1$ be a real number, $\sg$ a reservation, and $\psi$ the class reservation induced by $\sg$. 
Suppose for any index $\ell \in [L]$ and any machine $i \in C_\ell$, we have $\sum_{j \in J : \sg(j) \ni i} \E[P_{j,\psi(j),\ell}] \leq \beta \cdot s_i$. 
Then, 
\begin{equation} \label{eq:multiassignmentcostub}    
\cost(\sg) \leq O(\log m/\log \log m) \cdot \beta + \sum_{j \in \jobset} \bigl(\E[P_{j,\psi(j),0}] / \max_{i \in \sg(j)} s_i \bigr).
\end{equation}
\end{lemma}
\begin{proof}
Our upper bound on $\cost(\sg)$ is based on a particular choice of a $\sg$-\consistent assignment $\asgn : \jobset \to [m]$ which depends on the realizations $P$. 
We remark that the choice of machine $\asgn(j)$ depends only on $P_j$ and not on $P_{j^\prime}$ for $j^\prime \neq j$.
For any job $j \in \jobset$, the super-exceptional portion $P_{j,\psi(j),0}$ is always processed by the fastest machine $i \in \sg(j)$, which also belongs to $C_{\nxt(0;\psi(j))}$. 
Next, for any $j \in \jobset$ and index $\ell \in \psi(j)$, the $\ell$-truncated portion $P_{j,\psi(j),\ell}$ is always processed by the unique class-$\ell$ machine that has been reserved for $j$.  
By the above discussion, the (random) load on a class-$\ell$ machine $i$ can be split into two parts: (i)~the \emph{truncated} portion $\load^{\tr}(i)$, given by $\bigl(\sum_{j \in J^{\tr}_i} P_{j,\psi(j),\ell}\bigr)/s_i$, where $J^{\tr}_i = \{ j \in \jobset : i \in \sg(j) \}$; and (ii)~the \emph{exceptional} portion $\load^{\ex}(i)$, given by $\bigl(\sum_{j \in J^{\ex}_i} P_{j,\psi(j),0}\bigr)/s_i$, where $J^{\ex}_i = \{ j \in \jobset : i \in \sg(j) \cap C_{\nxt(0;\psi(j))}\}$. 

\smallskip

The expected value of the exceptional makespan can be bounded as follows:
\[
\E[\max_{i \in [m]} \load^{\ex}(i)] \leq \sum_{i \in [m]} \E[\load^{\ex}(i)] = \sum_{j \in \jobset} \frac{\E[P_{j,\psi(j),0}]}{\max_{i \in \sg(j)} s_i}.
\]
Next, we bound the expected value of the truncated makespan as follows: 
\[
\E[\max_{i \in [m]} \load^{\tr}(i)] \leq B + \sum_{i \in [m]} \E[\max(0,\load^{\tr}(i) - B)],
\]
where $B$ is a threshold to be chosen later so that $\E[\max(0,\load^{\tr}(i) - B)] = o(1/m)$ holds for every $i \in [m]$. 
To this end, fix some class-$\ell$ machine $i$. 
By definition, $\load^{\tr}(i) = \sum_{j : i \in \sg(j)} P_{j,\psi(j),\ell}/s_i$ is a sum of independent $[0,1]$-bounded random variables, whose mean $\E[\load^{\tr}(i)]$ is at most $\beta$ (recall the hypothesis). 
Using Lemma~\ref{lem:uppertailmass} with $B = (2\euler\beta) \cdot \log m/\log \log m$ and $\delta = B/\beta - 1$, we get:
\[
\E[\max(0,\load^{\tr}(i) - B)] \leq \left( \frac{\euler}{1+\delta} \right)^{B} \leq \Bigl(\frac{2 \log m}{\log \log m}\Bigr)^{-(2 \cdot \log m/\log \log m)} = o(1/m).
\]
Using triangle inequality for norms, we get:
\begin{multline*}
\cost(\sg) \leq \E\Bigl[ \bigl(\max_{i \in [m]} \load^{\tr}(i)\bigr) + \bigl(\max_{i \in [m]} \load^{\ex}(i)\bigr) \Bigr] \leq O\Bigl(\frac{\log m}{\log \log m}\Bigr) \cdot \beta + \sum_{j \in \jobset} \frac{\E[P_{j,\psi(j),0}]}{\max_{i \in \sg(j)} s_i}. \tag*{\qedhere}
\end{multline*}
\end{proof}

Putting everything together, we get the following approximation result for \slbr. 

\rellogthm*

\begin{remark}
For any $k = o(\log m/\log \log m)$, the analysis of the approximation ratio in \Cref{rellog-thm} is tight up to constant factors when comparing against the LP based lower bound from \Cref{subsec:lpoptk}. 
See \Cref{append:examples} for a concrete instance. 
\end{remark}

\subsection{\boldmath Bicriteria \texorpdfstring{$O(1)$}{O(1)}-Approximation Algorithm} \label{related-bicriteria}
In the previous section, we saw that converting a given class reservation $\psi$ to a reservation $\sg$ by choosing one machine per class that appears in $\psi(j)$ leads to an $O(\log m/\log \log m)$-approximation algorithm. 
This loss is unavoidable. 
Imagine a situation where some size class $\ell$ has $|C_\ell|/p$ jobs with identical distributions, each being $\s_\ell$ with probability $p$ and $0$ with probability $1-p$. Standard arguments show that an  optimal assignment of these jobs to machines in $C_\ell$ will give an expected makespan of $O(\log |C_\ell|/\log \log |C_\ell|)$, even though the omniscient optimum is a constant.  
In \Cref{sec:identical}, we saw that sampling two uniformly random machines is useful in obtaining approximation guarantees relative to the average load over a group of identical machines. 
We take the same approach here and sample two uniformly random machines from each class that appears in $\psi(j)$. 
This leads to a bicriteria approximation algorithm for \slbr. 

\relbicriteriathm*

This theorem follows immediately from the following main technical lemma, applied to an approximately optimal class reservation obtained from the previous subsection. 

\begin{lemma} \label{lem:twochoicesperclass}
Let $\beta \geq 1$ be a real number and $\psi : \jobset \to 2^{[L]} \sm \{\emptyset\}$ be a class reservation such that for any index $\ell \in [L]$, we have $\sum_{j \in \jobset : \psi(j) \ni \ell} \E[P_{j,\psi(j),\ell}] \leq \beta \cdot |C_{\ell}| \cdot \s_{\ell}$. 
Consider the random reservation $\sg : \jobset \to 2^{[m]} \sm \{\emptyset\}$ obtained as follows: for every job $j \in \jobset$ and every index $\ell \in \psi(j)$, we sample two machines $i_1(j,\ell),i_2(j,\ell)$ independently and uniformly from $C_\ell$, and set $\sg(j) = \bigcup_{\ell \in \psi(j)} \{i_1(j,\ell), i_2(j,\ell)\}$. 
Then, 
\begin{equation} \label{eq:twomultiassignmentcostub} 
\E_{\sg}[\cost(\sg)] \leq O(\beta) + \sum_{j \in \jobset} \frac{\E[P_{j,\psi(j),0}]}{ \s_{\nxt(0;\psi(j))}}.
\end{equation}
\end{lemma}

\begin{proof}
We follow the same high-level strategy as in Lemma~\ref{lem:costofonechoice}. 
We separately upper-bound the cost of the truncated and exceptional load vectors, which together imply an upper bound on the expected cost of $\sg$ over the randomness in $\sg$ and the $P_j$s. 
For each job $j \in \jobset$ and index $\ell \in \psi(j)$, the $\ell$-truncated portion $P_{j,\psi(j),\ell}$ is always processed by machine $i_1(j,\ell)$ or $i_2(j,\ell)$ in the second stage. 
The actual choice is decided after the job sizes are realized. 
The truncated load vector is induced by these truncated jobs, and the associated makespan objective will be bounded by a certain max-density function that we introduced in Lemma~\ref{lem:density}. 
The exceptional load vector arises from the super-exceptional portions.
Recall that for any job $j$, $\nxt(0;\psi(j))$ is the smallest index in $\psi(j)$ and, by definition, $\sg(j)$ always contains either one or two machines from class $C_{\nxt(0;\psi(j))}$. 
For any $j$, we will always process the super-exceptional portion $P_{j,\psi(j),0}$ on machine $i_1(j,\nxt(0;\psi(j)))$ in the second stage; note that this machine has speed $\s_{\nxt(0;\psi(j))}$. 
Like in Lemma~\ref{lem:costofonechoice}, the expected value (over the randomness in the $P_j$s) of the exceptional makespan can be charged to the second term in the RHS of 
\eqref{eq:twomultiassignmentcostub}, for every realization of $\sg$. 

We now focus on the cost arising from truncated jobs. 
In Lemma~\ref{lem:density}, we saw that for any $p \in \Rp^{\jobset}$, there is a $\sg$-\consistent assignment $\asgn$ satisfying: 
\[
\mksp(\asgn;p) \leq \denlb(p;\sg,s) \cdot \max_{j \in \jobset} |\sg(j)|, 
\]
where 
\[
\denlb(p;\sg,s) \coloneqq \max_{R \subseteq [m] : R \neq \emptyset} \frac{\sum_{j \in \jobset: \sg(j) \subseteq R} p_j}{\sum_{i \in R} s_i}.
\]
The problem with using the above result in a black-box fashion is that the approximation guarantee would be $O(\max_{j} |\psi(j)|) = O(k)$, which could be a super-constant.  
Since we will always be processing the $\ell$-truncated jobs on class-$\ell$ machines, the random reservation $\sg$ has a nice structure which can be exploited to circumvent this $O(k)$ loss.  
To this end, we define an extended job-set $\jobset^\prime \coloneqq \{(j,\ell) : j \in \jobset, \ell \in \psi(j)\}$ comprising all job-index pairs w.r.t. the class reservation $\psi$. 
Next, let $P^\prime \coloneqq (P_{j,\psi(j),\ell})_{(j,\ell) \in \jobset^\prime}$ denote the joint distribution of all truncated job variables.  
Note that for distinct indices $\ell,\ell^\prime \in \psi(j)$, both $P_{j,\psi(j),\ell}$ and $P_{j,\psi(j),\ell^\prime}$ cannot simultaneously realize a positive value. 
As a consequence, each job $j$ that realizes a positive size (in the second stage) is processed on a single machine from $\sg(j)$.  
However, the random variables $P_{j,\psi(j),\ell}$ and $P_{j^\prime,\psi(j^\prime),\ell^\prime}$ are independent whenever $j \neq j^\prime$. 
Last, we define a $2$-reservation $\sgp : \jobset^\prime \to 2^{[m]} \sm \{\emptyset\}$ on the expanded job-set by $\sgp(j,\ell) = \sg(j) \cap C_\ell$. 
Note that the distribution of $\sgp(j,\ell)$ is given by two independent and uniform samples from $C_\ell$. 
Using Lemma~\ref{lem:density}, we get the following (random) upper bound on the truncated cost of $\sg$ w.r.t. $P$:
\begin{align} \label{eq:expandedtwoassignment}
\cost(\sg) & = \E[\cost(\sg;P)] \leq \E[\cost(\sgp;P^\prime)] \notag \\ 
& \leq 2 \cdot \max_{R \subseteq [m]: R \neq \emptyset} \frac{\sum_{(j,\ell) \in \jobset^\prime} P_{j,\psi(j),\ell} \cdot \bone(\sgp(j,\ell) \subseteq R)}{\sum_{i \in R} s_i}.
\end{align} 
By construction, the $2$-reservation $\sgp$ has the nice property that for any $(j,\ell) \in \jobset^\prime$, we have $\sgp(j,\ell) \subseteq C_\ell$. 
Since the speed classes partition the machine-set $[m]$, Lemma~\ref{lem:density}(iii) implies that the maximum density in the RHS of \eqref{eq:expandedtwoassignment} is attained by a subset of one of the speed classes. 
Therefore, we get:
\begin{equation} \label{eq:costatmosttwicedensity}
\cost(\sg) \leq 2 \cdot \max_{\ell \in [L], R \subseteq C_\ell : R \neq \emptyset} \frac{\sum_{(j,\ell) \in \jobset^\prime} P_{j,\psi(j),\ell} \cdot \bone(\sgp(j,\ell) \subseteq R)}{|R| \cdot \s_\ell}  
\end{equation}

In the above inequality, there are two sources of randomness: $\sg$ and the $P_j$s. 
We first show that the expectation of $\cost(\sg)$ over the randomness in $\sg$ can be bounded by $O(\alpha)$, for some sufficiently large $\alpha$, whenever the event $\{\sum_{j \in \jobset : \psi(j) \ni \ell} P_{j,\psi(j),\ell} \leq \alpha \cdot |C_\ell| \cdot \s_\ell \, \forall \ell \in [L]\}$ happens (\Cref{boundedgamma}). 
We then use Chernoff bounds to show that $\E_P[\alpha] = O(\beta)$ (\Cref{alphaconcentration}), which completes the proof. 
To this end, consider a real-valued random function $\alpha$ that captures the maximum multiplicative gap between the total volume of the $\ell$-truncated job variables and the available processing power in class $C_\ell$ over all $\ell$: 
\[
\alpha(P) \coloneqq \max\Bigl(1,\max_{\ell \in [L]} \frac{\sum_{j \in \jobset : (j,\ell) \in \jobset^\prime} P_{j,\psi(j),\ell}}{|C_\ell| \cdot \s_\ell}\Bigr).
\]

The following two results follow from a slight generalization of the ideas we used in \Cref{thm:twochoicessuffice}. 

\begin{lemma} \label{boundedgamma}
Let $\gamma \geq 1$ be a real number and $p \in [0,1]^{\jobset^\prime}$. 
Suppose for all $\ell \in [L]$ we have $\sum_{j \in \jobset : (j,\ell) \in \jobset^\prime} p_{j,\ell} \leq \gamma \cdot |C_\ell|$. 
Then, 
\[
\E_{\sgp}\Bigl[ \max_{\ell \in [L], R \subseteq C_\ell : R \neq \emptyset} \frac{\sum_{j \in \jobset: (j,\ell) \in \jobset^\prime} p_{j,\ell} \cdot \bone(\sgp(j,\ell) \subseteq R)}{|R|} \Bigr] \leq \euler^2 \gamma + 2.
\]
\end{lemma}

\begin{proof}
For any index $\ell \in [L]$ and any nonempty $R \subseteq C_\ell$, let $X_R \coloneqq \sum_{j : (j,\ell) \in \jobset^\prime} p_{j,\ell} \cdot \bone(\sgp(j,\ell) \subseteq R)$ denote the random sum of $p_{j,\ell}$s for which the event $\{\sgp(j,\ell) \subseteq R\}$ happens. 
For any $(j,\ell) \in \jobset^\prime$, the distribution function of $\sgp(j,\ell)$ is just the outcome of sampling two machines independently and uniformly from $C_\ell$. Therefore, $\E[X_R] = \bigl(\sum_{j : (j,\ell) \in \jobset^\prime} p_{j,\ell}\bigr) \cdot \bigl(|R|^2/|C_\ell|^2\bigr) \leq \gamma \cdot |R|^2 / |C_\ell|$.
Using a union-bound calculation, we get: 
\[
\LHS \leq \euler^2 \gamma + \sum_{\ell \in [L], R \subseteq C_\ell : R \neq \emptyset} \E_{\sgp}\bigl[ \bigl( X_R - \euler^2 \gamma |R| \bigr)^{+} \bigr]/|R|.
\]
Using Lemma~\ref{lem:uppertailmass} with $B = \euler^2 \gamma |R|$ and $\delta = \euler^2 |C_\ell|/|R| - 1$, we get: 
\[
\E_{\sgp}\bigl[ \bigl( X_R - \euler^2 \gamma |R| \bigr)^{+} \bigr] \leq \Bigl(\frac{|R|}{\euler |C_\ell|}\Bigr)^{\euler^2 \gamma |R|}.
\]
Summing the above inequality over all subsets of some class $C_\ell$ and using $\binom{n}{k} \leq (\euler n/k)^k$ yields:
\begin{multline*}
\sum_{R \subseteq C_\ell : R \neq \emptyset} \frac{1}{|R|} \cdot \E_{\sgp}\bigl[ \bigl( X_R - \euler^2 \gamma |R| \bigr)^{+} \bigr] \leq \sum_{r=1}^{|C_\ell|} \binom{|C_\ell|}{r} \cdot \Bigl(\frac{r}{\euler |C_\ell|}\Bigr)^{\euler^2 \gamma r} \leq \sum_{r=1}^{|C_\ell|} \Bigl(\frac{r}{\euler |C_\ell|}\Bigr)^{6r} \leq \frac{1}{|C_\ell|}.
\end{multline*}
Summing the above inequality over all $\ell \in [L]$ yields: 
\[
\LHS \leq \euler^2 \gamma + \sum_{\ell = 1}^{L} |C_\ell|^{-1} \leq \euler^2 \gamma + 2. \tag*{\qedhere}
\]
\end{proof}

\begin{lemma} \label{alphaconcentration} 
For any $\gamma \geq \euler^2 \beta$, we have $\Pr[\alpha(P) > \gamma] \leq 2\euler^{-\gamma}$. 
\end{lemma}

\begin{proof}
Fix some $\gamma \geq \euler^2 \beta$.
For any $\ell \in [L]$, define $X_\ell \coloneqq \sum_{j \in \jobset : (j,\ell) \in \jobset^\prime} P_{j,\psi(j),\ell}/\s_\ell$. 
Note that $X_\ell$ is a sum of independent $[0,1]$-bounded random variables with mean $\E[X_\ell] \leq \beta |C_\ell|$ (recall the hypothesis of the lemma). 
Using Theorem~\ref{thm:chernoff} with $\mu = \beta |C_\ell|$ and $\delta = \euler^2-1$ gives: 
\[
\Pr[\alpha(P) > \gamma] \leq \sum_{\ell = 1}^{L} \Pr\bigl[ X_\ell > \gamma \cdot |C_\ell|\bigr] \leq \sum_{\ell = 1}^{L} \euler^{-\gamma \cdot |C_\ell|} \leq 2 \euler^{-\gamma}. \tag*{\qedhere} 
\]
\end{proof}

Putting everything together we get the following bound on the expected value of the truncated cost where we use $\tau = 2\euler^4 \beta + 4$:
\begin{flalign}
\E_{\sg,P}[\cost(\sg;P)] & \leq \tau + \int_{\tau}^{\infty} \Pr_{P}\Bigl[ \E_{\sgp}\bigl[\cost(\sgp;P^\prime)\bigr] > \gamma \Bigr] \, d\gamma \notag \\
& \leq \tau + \int_{\tau}^{\infty} \Pr_{P}\Bigl[ \E_{\sgp}\Bigl[ \max_{\ell \in [L], R \subseteq C_\ell : R \neq \emptyset} \frac{\sum_{j \in \jobset: (j,\ell) \in \jobset^\prime} P^\prime_{j,\ell} \cdot \bone(\sgp(j,\ell) \subseteq R)}{|R|} \Bigr] > \frac{\gamma}{2} \Bigr] \, d\gamma \tag{by \eqref{eq:costatmosttwicedensity}} \\
& \leq \tau + \int_{\tau}^{\infty} \Pr_{P}\Bigl[ \alpha(P) > \frac{\gamma-4}{2\euler^2} \Bigr] \, d\gamma \tag{contrapositive of \Cref{boundedgamma}} \\
& \leq \tau + 2 \cdot \int_{\tau}^{\infty} \euler^{-(\gamma-4)/(2\euler^2)} \, d\gamma \tag{\Cref{alphaconcentration}} \\
& \leq \tau + 4 \euler^2 \cdot \int_{\euler^2\beta}^{\infty} \euler^{-\gamma} \, d\gamma \notag \\
& = \tau + 4 \euler^2 \cdot \euler^{-\euler^2 \beta} \leq \tau + 1 = O(\beta). \notag
\end{flalign}

To summarize, the exceptional cost is at most $\sum_{j \in \jobset} (\E[P_{j,\psi(j),0}] / \s_{\nxt(0;\psi(j))})$ and the truncated cost is $O(\beta)$.  
\end{proof}

\section{Effectiveness of 2-Reservations versus Adaptive Optimum for Related Machines} \label{sec:adaptive}

Consider a set of related machines $M$ with speeds $\{s_i\}_{i\in M}$ and a set of jobs $\jobset$ with independent job-size distributions $\{P_j\}_{j \in \jobset}$.  
An adaptive policy is an algorithm that is allowed to adaptively process the jobs on the machines one by one. 
The processing time of a job is only observed after irrevocably committing to process it on a certain machine. 
Once its processing time has been observed, 
the algorithm can decide which of the remaining jobs to process next, and can assign it to any of the machines.
We say that an adaptive policy $\A$ with respect to $M$ and $\jobset$ is optimal if it minimizes the expected makespan, given by 
\[
\E\Bigl[ \max_{i \in M} \sum_{j \in \jobset} \frac{P_j}{s_i} \cdot \bone(j \mapsto i) \Bigr],
\]
where $\bone(j\mapsto i)$ is the indicator random variable for the event that $\A$ assigns job $j$ to machine $i$, and where the expectation is over the random processing times and all of the (possibly random) decisions that $\mathcal{A}$ makes.
We use $\OPTAD$ to denote the expected makespan of an optimal adaptive policy.
The main result of this section is the following.

\reladapthm*

\Cref{reladap-thm} implies that in instances where adaptive policies perform well with respect to the omniscient optimum $\OPT_m$, $2$-reservations should perform well too. 
For example, consider instances on identical machines. 
It is not hard to see that an adaptive policy given by a list schedule---one that processes the jobs in an arbitrary order, assigning each job to the current least-loaded machine---results in a schedule whose makespan is within twice the optimum for that realization of the job sizes. 
Therefore, up to constant factors, Theorem~\ref{reladap-thm} is a stronger version of Theorem~\ref{thm:iden2options}.  
Other types of instances for which adaptive policies perform well are those with weighted Bernoulli jobs, and arbitrary machine speeds. 
We say that $j \in \jobset$ is a weighted Bernoulli job if for some $w > 0$ and $q \in [0,1]$, $P_j$ takes size $w$ with probability $q$, and size $0$ with remaining probability.

\begin{corollary} 
Given a set $M$ of related machines and a set $\jobset$ of independent weighted Bernoulli jobs, one can efficiently compute a 2-reservation $\sg : \jobset \to 2^M$ such that $\E_{\sg}[\cost(\sigma)] = O(\OPT_m)$.
\end{corollary}

\begin{proof}
By Theorem~\ref{reladap-thm}, it suffices to show that $\OPTAD = O(\OPT_m)$. 
For each $j\in\jobset$, suppose that $P_j$ takes a positive size $w_j$ with probability $q_j$. 
Assume without loss of generality that the jobs are ordered so that $w_1\geq w_2\geq\hdots\geq w_n$, where $n=|\jobset|$. 
Consider the adaptive policy that processes the jobs in this order, assigning the current job $j$ to the machine $i$ that would finish processing $j$ the fastest, should $P_j$ instantiate to size $w_j$, i.e., 
\[
i=\operatorname*{argmin}_{i^\prime \in M}
\left(
\sum_{j^\prime < j}
\frac{p_{j^\prime}}{s_{i^\prime}}
\cdot \bone(j^\prime \mapsto i^\prime)
+
\frac{w_j}{s_{i^\prime}}
\right).
\]
It is a well-known fact that when job sizes are deterministic, that is, $q_j=1$ for all $j\in\jobset$, 
the above list schedule has a makespan that is at most twice the corresponding optimum. 
Note that for any realization $\{p_j\}_{j\in\jobset}$ of the processing times, jobs with positive $p_j$ are processed according to a deterministic list schedule by the adaptive policy. 
Therefore, conditioned on $\{p_j\}_{j\in\jobset}$, the resulting schedule is a 2-approximation for that particular realization.
It follows that the expected makespan of the adaptive policy is at most twice the omniscient optimum.  
\end{proof}

Our proof of \Cref{reladap-thm} is based on exploiting some structural results on near-optimal adaptive policies that Eberle et al. established in \cite{EberleGMMZ25}. 
We would like to reuse their results in a black-box fashion; however, their notion of \smoothing the machine speeds is slightly different. 
Thus, we have to show two things: (i)~the adaptive optimum does not increase by a super-constant factor under our \smoothing operation; and (ii)~our \smoothed machines satisfy the \smoothing properties that their results require. 
The second point can be checked by inspecting their machine \smoothing procedure (see Algorithm~6 in \cite{EberleGMMZ25}). 
The following lemma addresses the first point. 

\begin{lemma} \label{lem:adaptsmoothing} 
Suppose we are given an instance $\I = (M, (P_j)_{j \in \jobset})$ of stochastic load balancing on a set $M$ of related machines. 
We can efficiently compute a set of \smoothed machines $M^\prime$ with $|M^\prime| = O(|M|^2)$ such that $\OPTAD(P,M^\prime) = \Theta(\OPTAD(P,M))$. 
\end{lemma}

\begin{proof}
The proof heavily uses \Cref{lem:smoothing}. 
Construct the \smoothed{} collection $M^\prime$ according to \Cref{lem:smoothing}, and let $f : M \to M^\prime$ and $g : M^\prime \to M$ be the universal mappings with the claimed guarantees. 
Given any adaptive policy $\phi$ for the original instance, we can simulate an adaptive policy $\phi^\prime$ for the \smoothed instance as follows. 
Whenever policy $\phi$ decides to process job $j$ on machine $i \in M$, the policy $\phi^\prime$ processes this same job on machine $f(i)$. 
We look at the transcript once $\phi$ has processed all the jobs. 
Let $\asgn : \jobset \to M$ be the assignment that was made and let $\{p_j\}_{j \in \jobset}$ be the realized job sizes. 
By \Cref{lem:smoothing}~\ref{prop:smoothinglossforward}, we get that $\mksp(f \circ \asgn;p) \leq 12 \cdot \mksp(\asgn;p)$. 
Taking an expectation over the entire decision tree of $\phi$, we get that the expected makespan of the policy $\phi^\prime$ is at most $12$ times the expected makespan of $\phi$. 

The argument for the other direction is completely analogous by using the mapping $g$ and the reverse approximation loss given by  \Cref{lem:smoothing}~\ref{prop:smoothinglossreverse}. 

As the above argument works for any adaptive policy, we get that the adaptive optimum does not change by more than a constant factor under the \smoothing operation.  
\end{proof}

Both \Cref{lem:optundersmoothing,lem:adaptsmoothing} have the same underlying \smoothing procedure (given by \Cref{lem:smoothing}). 
So, up to constant factors, it suffices to establish Theorem~\ref{reladap-thm} under the assumption that ${M = C_1 \cup \ldots \cup C_L}$ is \smoothed and has $L$ geometric speed classes for some $L = O(\log |M|)$. 

\smallskip

We now formally state the structural result of Eberle et al. in the context of our problem. 

\begin{theorem}[Theorem~7 in \cite{EberleGMMZ25}]\label{thm:eberlestruct}
There exists an adaptive policy with expected makespan and total expected exceptional load at most $2 \cdot \OPTAD$ with respect to any truncation threshold $\tau \geq 2 \cdot \OPTAD$. 
Further, if job $j$ is assigned to machine $i$ by this policy then $\E[P_j] \leq \tau \cdot s_i$.
\end{theorem}

In the above theorem, the total expected exceptional load refers to the following quantity that depends on the truncation threshold $\tau$: 
\[
\E\Bigl[ \sum_{i \in M} \sum_{j \in \jobset} \frac{P_j}{s_i} \cdot \bone(P_j \geq \tau s_i) \cdot \bone(j \mapsto i) \Bigr].
\]
A nice consequence of \Cref{thm:eberlestruct} is that it gives an $O(1)$-approximate LP-based lower bound on $\OPTAD$. 
The LP, given by \eqref{eq:adapttrunc}--\eqref{eq:adaptnonneg}, 
has variables $y_{ij}$ that model the probability that job $j$ is processed on machine $i$ in a near-optimal adaptive policy (e.g., the policy in \Cref{thm:eberlestruct}). 
We refer to this LP as $\ADAPTLP(\tau)$, where $\tau$ is a threshold parameter. 

\begin{align} 
\sum_{j \in \jobset} \E\bigl[P_j \cdot \bone(P_j < \tau s_i) \bigr] \cdot y_{ij} & \, \, \leq \, \, \tau s_i && \forall i \in M \label{eq:adapttrunc} \\
\sum_{i \in M, j \in \jobset} \frac{\E\bigl[ P_j \cdot \bone(P_j \geq \tau s_i) \bigr]}{s_i} \cdot y_{ij} & \, \, \leq \, \, \tau && \label{eq:adaptexcep} \\
\sum_{i \in M} y_{ij} & \, \, = \, \, 1 && \forall j \in \jobset \label{eq:adaptassign} \\ 
y_{ij} & \, \, = \, \, 0 && \forall i,j : \, \E[P_j] > \tau \cdot s_i \label{eq:adaptmassive} \\
y & \, \, \geq \, \, 0 && \label{eq:adaptnonneg}
\end{align}

\begin{lemma}[Lemma~8 in \cite{EberleGMMZ25}] \label{lem:adaptlpfeasibility}
$\ADAPTLP(\tau)$ has a feasible solution for any $\tau \geq 2 \cdot \OPTAD$. 
\end{lemma}

We are now ready to prove the main result in this section. 

\begin{proof}[\Cref{reladap-thm}]
Using binary search, we may assume that we have a $\tau$ such that $\ADAPTLP(\tau)$ is feasible, whereas $\ADAPTLP(\tau/1.01)$ is not feasible. 
By \Cref{lem:adaptlpfeasibility}, $\tau \leq 2.02 \cdot \OPTAD$. 
We scale the $P_j$s by $\tau$ so that $\OPTAD = \Omega(1)$ and $\ADAPTLP(1)$ is feasible. 
It remains to compute a $2$-reservation whose expected cost is $O(1)$. 

Fix a feasible solution $\by$ to $\ADAPTLP(1)$. 
We now write another linear program by consolidating all $y_{ij}$ variables within each speed class $C_\ell$, i.e., for each job $j$ and class index $\ell$, we take $z_{\ell,j} = \sum_{i \in C_\ell} y_{ij}$. 
The volume constraint~\eqref{eq:adapttrunc} can be easily consolidated within any single class $C_\ell$ as $s_i = \s_\ell$ for all $i \in C_\ell$. 
We discard constraint~\eqref{eq:adaptmassive} as it is not needed for our purposes. 
We refer to the linear program given by \eqref{eq:adapttruncclass}--\eqref{eq:adaptnonnegclass} as $\ADAPTCLP$. 
 
\begin{align} 
\sum_{j \in \jobset} \E[P_j \cdot \bone(P_j < \s_\ell)] \cdot z_{\ell,j} & \, \, \leq \, \, |C_\ell| \cdot \s_\ell & \forall \ell \in [L] \label{eq:adapttruncclass} \\
\sum_{\ell \in [L], j \in \jobset} \bigl( \E\bigl[ P_j \cdot \bone(P_j \geq \s_\ell) \bigr] / \s_\ell \bigr) \cdot z_{\ell,j} & \, \, \leq \, \, 1 & \label{eq:adaptexcepclass} \\
\sum_{\ell \in [L]} z_{\ell,j} & \, \, = \, \, 1 & \forall j \in \jobset \label{eq:adaptassignclass} \\ 
z & \, \, \geq \, \, 0 & \label{eq:adaptnonnegclass}
\end{align}

Define $\bz \in [0,1]^{[L] \times \jobset}$ as follows: for each job $j$ and index $\ell$, let
\[
\bz_{\ell,j} \coloneqq \sum_{i \in C_\ell} \by_{ij}.
\]
By construction, $\bz$ is feasible to $\ADAPTCLP$.
We round $\bz$ directly using Theorem~\ref{thm:iterativerounding}.
Let $\U \coloneqq [L] \times \jobset$ and let $\Q \coloneqq \{ x \in \Rp^{\U} : \sum_{\ell \in [L]} x_{\ell,j} = 1 \, \forall j \in \jobset \}$.
The polytope $\Q$ is the base polytope of a partition matroid whose
parts are $\{(\ell,j):\ell\in[L]\}$ for $j\in\jobset$.
Its integral points are in one-to-one correspondence with integral
$1$-class reservations. 
For every $r,\ell \in [L]$ and $j \in \jobset$, define
\[
A_{r,\ell,j}
\coloneqq
\begin{cases}
\displaystyle
\frac{\E[P_j \cdot \bone(P_j < \s_\ell)]}
{|C_\ell| \cdot \s_\ell},
& \text{if } r=\ell,\\[8pt]
0,
& \text{otherwise}.
\end{cases}
\]
Also define
\[
c_{\ell,j} \coloneqq \frac{ \E[P_j \cdot \bone(P_j \geq \s_\ell)]}{\s_\ell}.
\]
Consider the multi-budgeted matroid LP given by: $
\min\{ c^\top x : Ax \leq \mathbf{1}, \ x \in \Q\}$. 
The point $\bz$ is feasible for this LP. Indeed,
constraint~\eqref{eq:adapttruncclass} implies that, for every
$r\in[L]$,
\[
(A\bz)_r
=
\frac{
\sum_{j \in \jobset}
\E[P_j \cdot \bone(P_j < \s_r)] \cdot \bz_{r,j}
}{
|C_r| \cdot \s_r
}
\leq 1,
\]
and constraint~\eqref{eq:adaptexcepclass} implies that
\[
c^\top \bz
=
\sum_{\ell \in [L],\,j \in \jobset}
\frac{\E[P_j \cdot \bone(P_j \geq \s_\ell)]}
{\s_\ell}
\cdot \bz_{\ell,j}
\leq 1.
\]
Moreover, every column of $A$ has sum at most $1$. In fact, for every
$(\ell,j)\in\U$,
\[
\sum_{r\in[L]} A_{r,\ell,j}
=
\frac{\E[P_j \cdot \bone(P_j < \s_\ell)]}
{|C_\ell| \cdot \s_\ell}
\leq
\frac{1}{|C_\ell|}
\leq 1.
\]
We can therefore invoke Theorem~\ref{thm:iterativerounding} with $\viol = 1$. 
It returns an integral point $\hat{z} \in \Q$ satisfying $c^\top \hat{z} \leq c^\top \bz \leq 1$ and $A\hat{z} \leq 2 \cdot \mathbf{1}$. 

Let $\psi : \jobset \to [L]$ be the integral $1$-class reservation
corresponding to $\hat{z}$; that is, $\psi(j) = \ell$ precisely when $\hat{z}_{\ell,j} = 1$. 
We then have
\[
\sum_{j \in \jobset}
\frac{\E[P_j \cdot \bone(P_j \geq \s_{\psi(j)})]}
{\s_{\psi(j)}}
=
c^\top \hat{z}
\leq 1,
\]
and, for every $\ell\in[L]$,
\[
\sum_{j \in \jobset:\,\psi(j)=\ell}
\E[P_j \cdot \bone(P_j < \s_\ell)]
\leq
2 \cdot |C_\ell| \cdot \s_\ell.
\]

Finally, we apply Lemma~\ref{lem:twochoicesperclass} to $\psi$ with
$\beta=2$ to obtain a random $2$-reservation $\sg$ satisfying
\[
\E_{\sg}[\cost(\sg)]
\leq
O(\beta)
+
\sum_{j \in \jobset}
\frac{\E[P_j \cdot \bone(P_j \geq \s_{\psi(j)})]}
{\s_{\psi(j)}}
=
O(1).
\tag*{\qedhere}
\]
\end{proof}

We conclude this section by showing that in some instances $2$-reservations can have significantly superior performance guarantees relative to the adaptive optimum. 

\begin{lemma} \label{lem:adaptvstworeslb}
For any $m \geq 1$, there is an instance of stochastic load balancing on $m+1$ related machines for which $\OPTAD/\OPT_2 = \Omega(\sqrt{m})$.
\end{lemma}
\begin{proof}
We describe an instance with two speed classes and $n \gg m$ independent and identically distributed jobs for which $\OPTAD = \Omega(\sqrt{m})$ and $\OPT_2 = O(1)$. 

We have one fast machine with speed $\sqrt{m}$ and $m$ slow machines with speed $1$.
For each job $j \in [n]$, its size $P_j$ is distributed as $m/n  + \sqrt{m} \cdot \Ber(1/n)$. 
Since any adaptive policy must risk processing at least $n/2$ jobs on the fast machine or (cumulatively over) the slow machines, $\OPTAD = \Omega(\sqrt{m})$ holds.   
This is because, in the former case, the total volume arising from the tiny deterministic portions $m/n$ causes the fast machine's load to exceed $\sqrt{m}/2$.
In the latter case, with constant probability, one of the $P_j$s assigned to a slow machine realizes a size of $m/n + \sqrt{m}$. 
This too leads to a load of $\sqrt{m}$ on the slow machine, and so $\OPTAD = \Omega(\sqrt{m})$. 

We claim that the following $2$-reservation $\sg$ has $O(1)$ cost: for each job $j \in [n]$, we always reserve the fast machine and the $\ceil{mj/n}$th slow machine for $j$. 
We process each $P_j$ on the fast machine if and only if it realizes a large size. 
Each slow machine appears in the reservation for roughly $n/m$ jobs, so even if all of them realize a small size, the total volume is $\approx 1$.
Thus, the load on every slow machine is always bounded by $1$. 
Since the expected number of large job-size realizations is $1$, we get that the expected load on the fast machine is also $O(1)$. 
Therefore, $\OPT_2 = O(1)$.  
\end{proof}

\section{Conclusions} \label{conclusions}

Our work opens up the study of reservation-based algorithms for stochastic load balancing as an exciting research direction, with many interesting avenues for exploration. 
We mention just a few questions here, and hope that our work will stimulate further work on this topic. 
Our results show that for identical and related machines, even $2$-reservations can offer a surprising amount of power. 
Obtaining a true $O(1)$-approximation for \slbr on related machines is still open.
We leave the setting of unrelated machines as a concrete open question for future work.
An intermediate setting is that of restricted machines, wherein each job can only be processed on some given subset of machines, but its processing time is the same on all machines in this set. 
It would be interesting to see if one can obtain guarantees here similar to those for identical machines. 
In particular, can a $2$-reservation yield a makespan that is within an $O(1)$ factor of
$\OPT_m$, or (if not) the adaptive optimum? 

Another question is whether our result for identical machines can be derandomized: can we efficiently compute a (deterministic) $2$-reservation with cost $O(\OPT_m)$?

It would also be interesting to incorporate the concept of $k$-reservations into other stochastic-load-balancing models, where varying $k$ would again allow one to interpolate between different solution concepts.  
For instance, while \slbr is a 2-stage model with offline stochastic jobs, one can also consider the setting where jobs are assigned adaptively with partial job-realization information at hand, and one needs to be \consistent with a $k$-reservation (chosen beforehand). 
This would interpolate between non-adaptive policies ($k=1$) and adaptive policies ($k=m$).

\paragraph{Acknowledgements}
David thanks Sepehr Assadi, and Sharat thanks Vipin Ravindran Vijayalakshmi for helpful discussions.

\printbibliography

\appendix

\section{Integrality gap of our linear relaxation for \boldmath \slbr} \label{append:examples} 

In this section, we show that the linear relaxation for \slbr discussed in \Cref{subsec:lpoptk} (see \eqref{eq:ellexceptionalvolume}) has an integrality gap of $\Theta(\log m/\log \log m)$.
Thus, improving upon the approximation guarantee in \Cref{rellog-thm} requires a stronger lower bound.  

The instance showcasing this gap is a slight modification of the \slbr instance given in \Cref{table:porlb}. 
The instance has $m+o(m)$ related machines that can be grouped in $k$ speed classes and there are $n \gg m$ i.i.d. jobs.   
In what follows, $k = o(\log m/\log \log m)$ and $\scaleparm = \Theta(\log m/\log \log m)$.

\begin{table}[ht]
\scriptsize 
\centering
\begin{equation*}
{\renewcommand{\arraystretch}{1.5}
\begin{array}{|c|c|c|c|c|c|c|c|c|} 
\hline
\text{(class index)} \, \ell & 1 & 2 & \ldots & \ell & \ldots & k-1 & k \\
\hline
\text{(\# class-} \ell 
\text{ machines)} \, |C_\ell| & 1 & \scaleparm^4 & \ldots & \scaleparm^{4\ell-4} & \ldots & \scaleparm^{4k-8} & m \\
\hline
\text{(speed of class } \ell) \, \s_\ell & \scaleparm^{2k-2} & \scaleparm^{2k-4} & \ldots & \scaleparm^{2(k-\ell)} & \ldots & \scaleparm^2 & 1 \\
\hline
\text{(processing power in class } \ell) \, |C_\ell| \s_\ell & \scaleparm^{2k-2} & \scaleparm^{2k} & \ldots & \scaleparm^{2k+2\ell-4} & \ldots & \scaleparm^{4k-6} & m \\
\hline 
\E[N_\ell] = \sum_{j \in \jobset} \Pr[P_j = \s_\ell] = n \Pr[X = \s_\ell] & 1 & \scaleparm^{4} & \ldots & \scaleparm^{4\ell-4} & \ldots & \scaleparm^{4k-8} & m \\
\hline 
\end{array}}
\end{equation*}
\caption{Integrality gap example for our LP relaxation for \slbr}
\label{table:badlpgap}
\end{table}

We give a rough sketch of our integrality gap claim by assuming that $k$ and $\scaleparm$ have been chosen carefully to satisfy $\scaleparm^{4k} = o(m^{0.01})$. 
The LP lower bound for the above instance is $O(1)$, which is realized by reserving all classes for all jobs.  
Next, we argue that $\OPT_k = \Omega(\log m/\log \log m)$.
Fix an optimal $k$-reservation $\mset$. 
We consider three cases.
\begin{enumerate}[(a)]
\item Let $J_0$ denote the subset of jobs that lack a reservation for class $C_k$. 
If $|J_0| = \Theta(n)$, then the cost of $\mset$ is clearly $\Omega(m^{0.99})$. 
This is because the total processing power in the first $k-1$ speed classes is $o(m^{0.01})$, while the expected total volume of jobs that realize a size of exactly $1$ is $\Theta(m)$. 

\item Let $J_2$ denote the subset of jobs that have at least two machines reserved from class $C_k$. 
If $|J_2| = \Theta(n)$, then the cost of $\mset$ is $\omega(\log m/\log \log m)$, as every job in $J_2$ is devoid of a reservation for at least one of the first $k-1$ speed classes. 
In more detail, consider the \slbr sub-instance $\I^\prime = (J_2, C_1 \cup \ldots \cup C_{k-1}, k-2)$. 
Note that $\I^\prime$ is similar to the instance in \Cref{table:porlb} with $\scaleparm$ replaced by $\scaleparm^2$. 
The appropriate restriction $\mset^\prime$ of $\mset$ w.r.t. $\I^\prime$ is a $(k-2)$-reservation. 
By mimicking the proof of \Cref{relporlb-thm}, we can show that $\cost_{\I^\prime}(\mset^\prime) = \Omega(\scaleparm^2/(k-1)) = \omega(\log m/\log \log m)$. 
Note that the two reservations in class $k$ (for jobs in $J_2$) are not of much help since even a single job of size at least $\scaleparm^2$ being processed on a class $k$ machine causes the makespan to be at least $\scaleparm^2$. 

\item Given the above two cases, we may assume that every job has exactly one reservation from class $k$ under $\mset$ and that these reservations are spread out uniformly over $C_k$.  
This only leads to a constant-factor loss in estimating $\cost(\mset)$. 
We now do a more careful version of the balls-and-bins analysis. 
Let $J^\prime_i$ denote the set of jobs with a reservation for machine $i \in C_k$. 
Note that the $J^\prime_i$s partition the job-set. 
Let $X_i := \sum_{j \in J^\prime_i} P_j \cdot \bone(P_j = 1)$. 
By our assumption, $X_i$ behaves like a Poisson variable with mean $1$ for sufficiently large $n$. 
With constant probability, one of the $X_i$s is $\Omega(\log m/\log \log m)$. 
If this entire volume of size-$1$ jobs is handled by the reserved machine $i$, then we get that the makespan is $\Omega(\log m/\log \log m)$. 
Suppose this is not the case, i.e., each class-$k$ machine only handles a load of up to $\tau = o(\log m/\log \log m)$. 
Since $\Pr[X_i \geq \tau] \approx 1/\tau!$, we get that $\sum_{i \in C_k} \E[\max(0,X_i-\tau)] =  \Omega(m/\tau!) = \Omega(m^{0.99})$. 
Since the first $k-1$ speed classes have a total processing power of $o(m^{0.01})$, this excess load of $\Omega(m^{0.99})$ arising from machines in $C_k$ not handling loads above $\tau$ leads to a makespan of $\Omega(m^{0.98})$.
\end{enumerate}

\end{document}